\documentclass[11pt,reqno]{amsart}
\usepackage[T1]{fontenc}
\usepackage[utf8]{inputenc}
\usepackage[a4paper,margin=3.5cm]{geometry}
\usepackage[pagebackref,breaklinks,colorlinks=true,linkcolor=blue,citecolor=blue]{hyperref}
\usepackage{natbib,lmodern,mathtools}
\title[Concentration and transport for Coulomb gases]%
{Concentration~for~Coulomb~gases and~Coulomb~transport~inequalities}
\thanks{This work was supported in part by the Labex CEMPI (ANR-11-LABX-0007-01)}
\author{Djalil~Chafaï}%
\address[Djalil Chafa\"i]{Universit\'e Paris-Dauphine, PSL Research Unversity, France}%
\email{djalil@chafai.net}%
\author{Adrien~Hardy}%
\address[Adrien Hardy]{Université de Lille 1, France}%
\email{adrien.hardy@math.univ-lille1.fr}%
\author{Mylène~Maïda}%
\address[Mylène Maïda]{Université de Lille 1, France}%
\email{mylene.maida@math.univ-lille1.fr}%
\date{Summer 2017. Compiled \today}
\numberwithin{equation}{section}
\keywords{Coulomb gas; Ginibre ensemble; Random matrix; Transport of measure;
  Transport inequality; Talagrand inequality; Wasserstein distance;
  Concentration of measure.}
\subjclass[2000]{%
26D10; 
31A;   
31B;   
28A35; 
60E15; 
82B44; 
28C15} 
\newtheorem{theorem}{Theorem}[section]%
\newtheorem{corollary}[theorem]{Corollary}%
\newtheorem{lemma}[theorem]{Lemma}%
\newtheorem{remark}[theorem]{Remark}%
\newcommand{\dC}{\mathbb{C}}

\newcommand{\dP}{\mathbb{P}}
\newcommand{\dR}{\mathbb{R}}
\newcommand{\dS}{\mathbb{S}}

\newcommand{\cC}{\mathcal{C}}
\newcommand{\cE}{\mathcal{E}}

\newcommand{\cP}{\mathcal{P}}

\newcommand{\bH}{\mathbf{H}}

\newcommand{\bT}{\mathbf{T}}

\newcommand{\dd}{\mathrm{d}} 
\newcommand{\e}{\mathrm{e}} 
\newcommand{\ds}{\displaystyle}

\newcommand{\be}{\beta}
\newcommand{\De}{\Delta}
\newcommand{\de}{\delta}
\newcommand{\ga}{\gamma}
\newcommand{\Ga}{\Gamma}

\newcommand{\la}{\lambda}

\newcommand{\si}{\sigma}

\newcommand{\te}{\theta}

\newcommand{\R}{\dR}
\newcommand{\veps}{\varepsilon}
\newcommand{\vphi}{\varphi}

\DeclareMathOperator{\DIST}{\mathrm{d}}

\DeclareMathOperator{\ENE}{\cE}
\DeclareMathOperator{\ENTH}{H} 
\DeclareMathOperator{\ENTS}{S} 
\DeclareMathOperator{\HESS}{Hess} 

\DeclareMathOperator{\SUPP}{supp} 
\DeclareMathOperator{\TR}{Tr} 

\DeclareMathOperator{\VOL}{vol} 
\DeclareMathOperator{\WAS}{W} 
\newcommand{\ABS}[1]{{{\left| #1 \right|}}} 
\newcommand{\ANG}[1]{{{\langle#1\rangle}}} 
\newcommand{\LIP}[1]{\|#1\|_{\mathrm{Lip}}} 
\newcommand{\NRM}[1]{{{\left\| #1\right\|}}} 
\newcommand{\PAR}[1]{{{\left(#1\right)}}} 
\newcommand{\pd}{{\partial}} 
\newcommand{\IND}{\boldsymbol{1}}

\newcommand{\BL}{{\mathrm{BL}}}
\newcommand{\DBL}{\DIST_{\BL}}
\makeatletter
\def\@MRExtract#1 #2!{#1}     
\renewcommand{\MR}[1]{
  \xdef\@MRSTRIP{\@MRExtract#1 !}%
  \href{http://www.ams.org/mathscinet-getitem?mr=\@MRSTRIP}{MR-\@MRSTRIP}}
\makeatother
\begin{document}

\begin{abstract}
  We study the non-asymptotic behavior of Coulomb gases in dimension two and
  more. Such gases are modeled by an exchangeable Boltzmann--Gibbs measure
  with a singular two-body interaction. We obtain concentration of measure
  inequalities for the empirical distribution of such gases around their
  equilibrium measure, with respect to bounded Lipschitz and Wasserstein
  distances. This implies macroscopic as well as mesoscopic convergence in
  such distances. In particular, we improve the concentration inequalities
  known for the empirical spectral distribution of Ginibre random matrices.
  Our approach is remarkably simple and bypasses the use of renormalized
  energy. It crucially relies on new inequalities between probability metrics,
  including Coulomb transport inequalities which can be of independent
  interest. Our work is inspired by the one of Maïda and Maurel-Segala, itself
  inspired by large deviations techniques. Our approach allows to recover,
  extend, and simplify previous results by Rougerie and Serfaty.
\end{abstract}

\maketitle


\section{Introduction}

The aim of this work is the non-asymptotic study of Coulomb gases around their
equilibrium measure. We start by recalling some essential aspects of
electrostatics, including the notion of equilibrium measure. We then
incorporate randomness and present the Coulomb gas model followed by the
natural question of concentration of measure for its empirical measure. This
leads us to the study of inequalities between probability metrics, including
new Coulomb transport type inequalities. We then state our results on
concentration of measure and their applications and ramifications. We close
the introduction with some additional notes and comments.

In all this work, we take $d\geq2$. 

\subsection{Electrostatics}

The $d$-dimensional \emph{Coulomb kernel} is defined by
\[
x\in\dR^d\mapsto
g(x):=
\begin{cases}
  \ds \log\frac{1}{\ABS{x}}
  &\text{if $d=2$,}\\
  \ds\frac{1}{\ABS{x}^{d-2}}
  &\text{if $d\ge 3$}.
\end{cases}
\]
When $d = 3$, up to a multiplicative constant, $g(x)$ is the electrostatic
potential at $ x \in \R^3$ generated by a unit charge at the origin according
to Coulomb's law. More generally, it is the fundamental solution of Poisson's
equation. Namely, according to \cite[Th.~6.20]{lieb-loss}, if we denote by
$\De:=\pd_1^2+\cdots+\pd_d^2$ the Laplace operator on $\dR^d$ and by $\de_0$
the Dirac mass at the origin then, in the sense of Schwartz distributions,
\begin{equation}\label{eq:Poisson}
  \De g=-c_d\, \de_0,
\end{equation}
where $c_d$ is a positive  constant given by 
\[
c_d 
:=
\begin{cases}
  \ds 2\pi 
  &\text{if $d=2$},\\
  \ds (d-2)|\dS^{d-1}| 
  &\text{if $d\ge 3$},
\end{cases}\quad 
\mbox{ with } \quad |\dS^{d-1}|:=\frac{2\pi^{d/2}}{\Ga(d/2)}.
\]
The function $g$ is superharmonic on $\dR^d$, harmonic on
$\dR^d\setminus\{0\}$, and belongs to the space $ L^1_{\mathrm{loc}}(\R^d)$ of
locally Lebesgue-integrable functions.

Let $\cP(\dR^d)$ be the space of probability mesures on $\R^d$. For any
$\mu\in\cP(\dR^d)$ with compact support, its \emph{Coulomb energy},
\begin{equation}\label{eq:CE}
  \ENE(\mu):=\iint g(x-y)\mu(\dd x)\mu(\dd y)\in \R\cup\{+\infty\},
\end{equation}
is well defined since $g$ is bounded from below on any compact subset of
$\R^d$ (actually we even have $g\geq0$ when $d\geq3$). Recall that a closed
subset of $\dR^d$ has \emph{positive capacity} when it carries a compactly
supported probability measure with finite Coulomb energy; otherwise it has
\emph{null capacity}. We refer for instance to \citep{helms,MR0350027} for
$d\geq2$ and \citep{saff-totik} for $d=2$. 

The quantity $\ENE(\mu)$ represents the electrostatic energy of the
distribution $\mu$ of charged particles in $\dR^d$. Moreover, given any
$\mu,\nu\in \cP(\dR^d)$ with compact support and finite Coulomb energy, the
quantity $\ENE(\mu-\nu)$ is well defined, finite, non-negative, and vanishes
if and only if $\mu=\nu$, see \cite[Lem.~I.1.8]{saff-totik} for $d=2$ and
\cite[Th.~9.8]{lieb-loss} for $d\ge 3$. In particular, for such probability
measures, the following map is a metric:
\begin{equation}\label{eq:enedist}
  (\mu,\nu)\mapsto\ENE(\mu-\nu)^{1/2}.
\end{equation}

One can confine a distribution of charges by using an external potential. More
precisely, an \emph{admissible (external) potential} is a map
$V:\R^d\to\R\cup\{+\infty\}$ such that:
\begin{enumerate}
\item[1.] $V$ is lower semicontinuous;
\item[2.] $V$ is finite on a set of positive capacity;
\item[3.] $V$ satisfies the growth condition 
$$\lim_{|x|\to\infty}\big(V(x)-2\log|x|\,\IND_{d=2}\big)=+\infty.$$
\end{enumerate}
Now, for an admissible potential $V$, we consider the functional 
\begin{equation}\label{eq:EV}
  \ENE_V(\mu)
  :=\iint \left(g(x-y) + \frac12V(x)+\frac12V(y)\right)\mu(\dd x)\mu(\dd y).
\end{equation}
Note that, thanks to the assumptions on $V$, the integrand of \eqref{eq:EV} is
bounded from below. Thus $ \ENE_V(\mu)$ is well defined for every
$\mu\in\cP(\R^d)$, taking values in $\R \cup \{+\infty\}$. Moreover, when
$\mu\in\cP(\R^d)$ with $d\ge 3$ or $d=2$ and $\int_{\R^2} \log(1+|x|)\mu(\dd
x)<\infty$, the Coulomb energy $\ENE(\mu)$ is still well defined and in this
case one has
\begin{equation}\label{eq:EVEV}
  \ENE_V(\mu)= \ENE(\mu)+\int  V(x)\mu(\dd x).
\end{equation}
It is known that if $V$ is admissible then $\ENE_V$ has a unique minimizer
$\mu_V$ on $\cP(\dR^d)$ called the \emph{equilibrium measure}, which is
compactly supported. Moreover, if $V$ has a Lipschitz continuous derivative,
then $\mu_V$ has a density given by
\begin{equation}\label{eq:rhov}
  \rho_V=\frac{\De V}{2c_d}
\end{equation}
on the interior of its support. For example when $V=\ABS{\cdot}^2$ this yields
together with the rotational invariance that the probability measure $\mu_V$
is uniform on a ball. See for instance \citep[Prop.~2.13]{MR2647570},
\citep{chafai-gozlan-zitt}, and \cite[Sec.~2.3]{serfaty-course}.

\subsection{Coulomb gas model}

Given $N\geq 2$ unit charges $x_1,\ldots,x_N\in\R^d$ in a
confining potential $V$, the energy of the configuration is given by
\begin{equation}\label{eq:HN}
  H_N(x_1,\ldots,x_N) := \sum_{i\neq j} g(x_i-x_j) + N\sum_{i=1}^N V(x_i).
\end{equation}
The prefactor $N$ in front of the potential $V$ turns out to be the
appropriate scaling in the large $N$ limit, because of the long range nature
of the Coulomb interaction. The corresponding ensemble in statistical physics
is the Boltzmann--Gibbs probability measure on $(\dR^d)^N$ given by
\begin{equation}
\label{eq:defPN}
\dd \dP_{V,\be}^N(x_1,\ldots,x_N)
:= \frac{1}{Z^N_{V,\be}}\exp\bigg(-\frac{\be}{2}H_N(x_1,\ldots,x_N)\bigg)
\dd x_1\cdots \dd x_N,
\end{equation}
where $\be>0$ is the inverse temperature parameter and where $Z^N_{V,\be}$ is
a normalizing constant known as the partition function. To ensure that the
model is well defined, we assume that $V$ is finite on a set of positive
Lebesgue measure and
\begin{equation*}\label{eq:gamma}\tag{$\mathbf H_\beta$}
 \int \e^{-\frac\be 2 (V(x) - \mathbf 1_{d=2} \log(1+|x|^2))}\dd x <\infty.
\end{equation*}
This indeed implies that $0<Z^N_{V,\be}<\infty$. The \emph{Coulomb gas} or
\emph{one-component plasma} stands for the exchangeable random particles
$x_1,\ldots,x_N$ living in $\R^d$ with joint law $\dP_{V,\be}^N$. We refer to
\citep{MR2641363,serfaty-course} and references therein for more mathematical
and physical aspects of such models.

It is customary to encode the particles $x_1,\ldots,x_N$ by their empirical
measure
\[ 
\hat \mu_N := \frac{1}{N} \sum_{i=1}^N\de_{x_i}.
\]
It is known, at least when $V$ is admissible and continuous, that for any
$\beta>0$ satisfying \eqref{eq:gamma} we have the weak convergence, with
respect to continuous and bounded test functions,
\[
\hat\mu_N \underset{N\to\infty}{\longrightarrow} \mu_V,
\]
with probability one, in any joint probability space. This weak convergence
follows from the $\Ga$-convergence of $\frac1{N^2} H_N$ towards $\ENE_V$, or
from the large deviation principle at speed $N^2$ and rate function
$\frac{\be}{2}(\ENE_V-\ENE_V(\mu_V))$ satisfied by $\hat\mu_N$, see
\citep{serfaty-course,chafai-gozlan-zitt}. More precisely, the large deviation principle yields the convergence, for any distance $\dd$ on $\cP(\R^d)$ metrizing the weak topology  and 
$r>0$,
\[
\frac{1}{N^2}\log\dP_{V,\be}^N\Big(\,\dd(\hat\mu_N,\mu_V) \ge r\Big)%
\underset{N\to\infty}{\longrightarrow}
-\frac\be2\inf_{\substack{\dd(\mu,\mu_V)\ge r}}
  \big(\ENE_V(\mu)-\ENE_V(\mu_V)\big).
\] 
In particular, there exist non-explicit constants
$N_0,c,C>0$ depending on $\beta$, $V$ and $r>0$ such that, for every
$N\geq N_0$,
\begin{equation}
 \label{eq:LDP}
 \e^{- cN^2}  \leq \dP_{V,\be}^N \Big(\,\dd(\hat\mu_N,\mu_V)\ge r \Big)\leq \e^{- CN^2}.
\end{equation} 
Beyond this large deviation principle, one may look for concentration
inequalities, quantifying non-asymptotically the closeness of $\hat\mu_N$ to
$\mu_V$. Namely, can we improve the large deviation upper bound in
\eqref{eq:LDP} into an inequality holding for any $N\geq 2$ and with an
explicit dependence in $r$? A quadratic dependence of $C$ in $r$ would yield a
sub-Gaussian concentration inequality. Taking for $\dd$ stronger distances,
such as the Wasserstein distances $\WAS_p$ would be of interest too. For the
Coulomb gas with $d=2$ but restricted to the real line, also known as the
\emph{one-dimensional log-gas} popularized by random matrix theory (see
\citep{MR2641363} and references therein), a concentration inequality for
$\WAS_1$ has been established by \citet{maida-maurel-segala}. But for the
usual Coulomb gases in dimension $d\geq 2$, no concentration inequality is
available in the literature yet for metrics yielding the weak convergence.

Let us also mention that central limit theorems have been obtained for one
dimensional log-gases by \citet{johansson}, and for two dimensional Coulomb
gases by \citet{MR2361453,MR2817648,MR3342661} when $\beta=2$ and recently
extended by \citet{leble-serfaty-clt,bourgade-CLT} to any $\beta>0$. As for
the higher dimensional Coulomb gases, the fluctuations remain largely
unexplored.

A related topic is the meso/microscopic study of Coulomb gases, which has been
recently investigated by several authors including \citet{MR3353821,bourgade,
  leble} for $d=2$, \citet{MR3455593,petrache-serfaty,leble-serfaty} in any
dimension. These works provide fine asymptotics for local observables of
Coulomb gases, most of them using the machinery of the renormalized energy.
For our purpose, an interesting deviation bound can be derived from
\citep{MR3455593}: Under extra regularity conditions for $V$, for any $d\geq
2$ and any fixed $r,R>0$,
\begin{equation}
\label{eq:RougeSer}
\dP_{V,\be}^N\left(\sup_{\LIP{f}\leq1}\int_{|x|\leq R} f(x)(\hat\mu_N-\mu_V)(\dd x)\ge r \right)\leq \e^{- C_R\,r^2N^{2}}
\end{equation}
provided $N\geq N_0$, for some constants $C_R,N_0>0$ depending on $R,\beta,V$;
we used the notation 
\[
\LIP{f} := \sup_{x\neq y}\frac{\ABS{f(x)-f(y)}}{\ABS{x-y}}.
\]
Indeed, getting rid of the constraint ${|x|\leq R}$ in the integral and of the
dependence in $R$ of the constant $C_R$ would yield a concentration inequality
in the $\WAS_1$ metric, and this is exactly what we are aiming at. Although
one could try to adapt \citep{MR3455593}'s proof so as to extend the bound
\eqref{eq:RougeSer} to a non-local setting, we shall proceed here differently,
without using the electric field machinery developed in the context of the
renormalized energy. Our goal is to provide self-contained, short and simple
proofs for concentration inequalities of this flavor, as well as providing
more explicit constants and ranges of validity when it is possible.

More precisely in this work we provide, under appropriate regularity and
growth conditions on $V$, concentration inequalities for the empirical measure
of Coulomb gases in any dimension $d\geq 2$, with respect to both bounded
Lipschitz and Wasserstein $\WAS_1$ metrics. For both metrics, we obtain
sub-Gaussian concentration inequalities of optimal order with respect to $N$,
see Theorem \ref{th:theococo} and Theorem \ref{th:theococo2}. Our
concentration inequalities are new for general Coulomb gases and Corollary
\ref{co:ginibre} even improves the previous concentration bounds for the
empirical spectral measure of the Ginibre random matrices, due to
\cite{pemantle-peres} and \cite{breuer-duits}, both results relying on the
determinantal structure of the Ginibre ensemble. These inequalities are
precise enough to yield weak and $\WAS_1$ convergence of $\hat\mu_N$ to
$\mu_V$ at mesoscopic scales, see Corollary \ref{cor:lolo}. For completeness,
we also provide estimates on the probability of having particles outside of
arbitrarily large compact sets, see Theorem \ref{th:tightness}.



As for the proofs, a key observation is that the energy \eqref{eq:HN} is a function of the
empirical measure $\hat\mu_N$. Indeed, with the identity \eqref{eq:EVEV} in
mind, we can write
\begin{equation}\label{eq:cgmu}
  \dd\dP_{V,\be}^N(x_1,\ldots,x_N)%
  =\frac{1}{Z_{V,\be}^N}
  \exp\PAR{-\frac{\be}{2}N^2\ENE_V^{\neq}(\hat\mu_N)}\dd x_1\cdots\dd x_N
\end{equation}
where 
\begin{equation}\label{eq:enevneq}
  \ENE_V^{\neq}(\hat\mu_N):=\iint_{x\neq y} g(x-y)\hat\mu_N(\dd x)\hat\mu_N(\dd y) +\int V(x)\hat\mu_N(\dd x).
\end{equation}
The quantity $\ENE_V^{\neq}(\hat\mu_N)-\ENE_V(\mu_V)$ is an
energetic way to measure the closeness of $\hat\mu_N$ to $\mu_V$. One may alternatively write it as a regularized $L^2$ norm of the electric field of $\hat\mu_N-\mu_V$, which leads to the notion of  renormalized energy. As mentioned previously, it has been successfully used to obtain fine local asymptotics for Coulomb gases.  With a concentration inequality involving global observables as a target, we bypass the use of this machinery by using functional inequalities, in the spirit of Talagrand transport inequalities,  replacing the usual relative entropy by the Coulomb energy $\cE_V$ or
the Coulomb metric \eqref{eq:enedist}. Such inequalities, which we introduce now, may be of independent interest.

\subsection{Probability metrics and Coulomb transport inequality} 
In the following, we consider the \emph{bounded Lipschitz} (Fortet--Mourier) distance on $\cP(\R^d)$ defined by
\begin{equation}\label{eq:bl}
  \DBL(\mu,\nu)%
  :=\sup_{\substack{\LIP{f}\leq1\\\NRM{f}_\infty\leq1}}\int f(x)(\mu-\nu)(\dd x),
\end{equation}
which metrizes its weak convergence, see \citep{Dud02}. Moreover, for any $p\geq1$, the \emph{Kantorovich} or \emph{Wasserstein distance} of order $p$ between
$\mu\in\cP(\dR^d)$ and $\nu\in\cP(\dR^d)$ is defined by
\begin{equation}\label{eq:Wp}
  \WAS_p(\mu,\nu):=\Bigr(\inf\iint\ABS{x-y}^p\pi(\dd x,\dd y)\Bigr)^{1/p},
\end{equation}
where the infimum runs over all probability measures
$\pi\in\cP(\dR^d\times\dR^d)$ with marginal distributions $\mu$ and $\nu$,
when $\ABS{\cdot}^p$ is integrable for both $\mu$ and $\nu$. We set
$\WAS_p(\mu,\nu)=+\infty$ otherwise. Convergence in the $\WAS_p$ metric is
equivalent to weak convergence plus convergence of moments up to order $p$.
The Kantorovich-Rubinstein dual representation of $\WAS_1$
\citep{MR1964483,MR1619170} reads
\begin{equation}\label{eq:KRduality}
  \WAS_1(\mu,\nu)=\sup_{\LIP{f}\le 1}\int f(x)(\mu-\nu)(\dd x).
\end{equation} 
In particular, we readily have, for any $\mu,\nu\in\cP(\R^d)$,
\begin{equation}
\label{BL-W1 dominance}
  \DBL(\mu,\nu)\leq \WAS_1(\mu,\nu).
\end{equation}
Our first result states that the Coulomb metric \eqref{eq:enedist} locally
dominates the $\WAS_1$ metric, and hence the bounded Lipschitz metric as well.
In particular, once restricted to probability measures on a prescribed compact
set, the convergence in Coulomb distance implies the convergence in
Wasserstein distance.

\begin{theorem}[Local Coulomb transport inequality]%
  \label{th:clé}%
  For every compact subset $D\subset~\R^d$, there exists a constant $C_D>0$
  such that, for every $\mu,\nu\in\cP(\R^d)$ supported in $D$ with  
  $\ENE(\mu)<\infty$ and $\ENE(\nu)<\infty$,
  \[
  \WAS_1(\mu,\nu)^2\le C_D\, \ENE(\mu-\nu).
  \]
\end{theorem}

We provide a proof for Theorem \ref{th:clé} in Section \ref{se:th:clé}.

As we will see in the proof, a rough upper bound on the optimal constant $C_D$
is given by the volume of the ball of radius four times the one of the domain
$D.$

When $d=2$ and $D\subset\R$, Theorem \ref{th:clé} yields
\cite[Th.~1]{popescu}, up to the sharpness of the constant; we can say that
Theorem \ref{th:clé} extends Popescu's local free transport inequality to
higher dimensions.

How to go over the case of compactly supported measures and get rid of the
dependence in the domain $D$ of the constant $C_D$? To avoid restricting the
distribution of charges $\mu$ to a prescribed bounded domain $D$, we add a
confining external potential to the energy, namely we consider the weighted
functional $\ENE_V$ introduced in \eqref{eq:EV} instead of $\ENE$. This will
provide inequalities when one of the measure is the equilibrium measure
$\mu_V$; this is not very restrictive as any compactly supported probability
measure such that $g*\mu$ is continuous is the equilibrium measure of a well
chosen potential, see for instance \citep{MR0350027}. Our next result is a
Coulomb transport type inequality for equilibrium measures.

\begin{theorem}[Coulomb transport inequality for equilibrium measures]%
  \label{th:coultransp}%
  If the potential $V:\R^d\to\R\cup\{+\infty\}$ is admissible, 
  then there exists a constant $C^V_{\BL}>0$ such that, for every
  $\mu\in\cP(\R^d)$,
  \begin{equation}\label{eq:CoulombBL}
    \DBL(\mu,\mu_V)^2\le C_{\BL}^V\Big(\ENE_V(\mu)-\ENE_V(\mu_V)\Big).
  \end{equation}
  If one further assumes that $V$ grows at least quadratically, 
  \begin{equation*}\label{eq:Vgrowth}\tag{$\bH_{\WAS_1}$}
  \liminf_{|x|\to\infty}\frac{V(x)}{|x|^2}>0,
  \end{equation*}
  then there exists a constant $C^V_{\WAS_1}>0$ such that, for every
  $\mu\in\cP(\R^d)$,
  \begin{equation}\label{eq:CoulombW1}
    \WAS_1(\mu,\mu_V)^2\le C_{\WAS_1}^V\big(\ENE_V(\mu)-\ENE_V(\mu_V)\big).
  \end{equation}
\end{theorem}

Theorem \ref{th:coultransp} is proved in Section \ref{se:th:coultransp}.

Given a potential $V,$ one can obtain rough upper bounds on the optimal
constants $C_{\BL}^V,C^V_{\WAS_1}$ by following carefully the proof of the
theorem.


\begin{remark}[Free transport inequalities]
  A useful observation is that the problem of minimizing $\ENE_V(\mu)$ over
  $\mu\in \cP(S)$, where $S\subset\R^d$ has positive capacity, is equivalent
  to consider the full minimization problem but setting $V=+\infty$ on
  $\R^d\setminus S$. In particular, by taking any admissible $V$ on $\R$ and
  setting $V=+\infty$ on $\R^2\setminus \R$, Theorem \ref{th:coultransp} with
  $d=2$ yields the free transport inequality obtained by
  \citet{maida-maurel-segala}, without the restriction that $V$ is continuous,
  and \citet{popescu}. Notice that although the lower semicontinuity of $V$ is
  not assumed in \citep{popescu}, it is actually a necessary condition for the
  existence of $\mu_V$ in general.
\end{remark}

\begin{remark}[Optimality of the growth condition for $\WAS_1$]
  Following \cite[Rem.~1]{maida-maurel-segala}, one can check that
  \eqref{eq:CoulombW1} cannot hold if $V$ is admissible but does not satisfy
  \eqref{eq:Vgrowth}. Indeed, letting $\nu_n$ be the uniform probability
  measure on the ball of $\R^d$ of radius one centered at $(n,0,\ldots,0)$,
  then $\WAS_1(\nu_n,\mu_V)^2$ grows like $n^2$ as $n\to\infty$, whereas
  $\ENE_V(\nu_n)$ grows like $\int V\dd\nu_n=o(n^2)$.
\end{remark}

Equipped with these metric inequalities, we are now in position to state the
concentration inequalities.

\subsection{Concentration of measure for Coulomb gases}

Our first result on concentration of measure is the following. 

\begin{theorem}[Concentration of measure for Coulomb gases]\label{th:theococo}%
  Assume that $V$ is $\cC^2$ on $\R^d$ and that its Laplacian $\Delta V$
  satisfies the following growth constraint
  \begin{equation}\label{eq:weakLapGrowth}
    \limsup_{|x|\to\infty}\Bigg(\frac1{V(x)} 
    \displaystyle\sup_{\substack{y\in\dR^d\\|y-x| < 1}}\Delta V(y)\Bigg)< 2(d+2).
  \end{equation}
  If $V$ is admissible then there exist constants $a>0$, $b\in\R$, and a
  function $\beta\mapsto c(\beta)$ such that, for any $\beta>0$ which
  satisfies \eqref{eq:gamma}, any $N\geq 2$, and any $r>0$,
  \begin{equation}\label{eq:concentrationBL}
    \dP_{V,\be}^N \Big(\DBL(\hat\mu_N,\mu_V)\ge r \Big) %
    \leq \e^{-a\beta N^2 r^2 %
      +\mathbf 1_{d=2}(\frac{\be}{4} N \log N) %
      +b\beta N^{2-2/d} %
      +c(\beta) N}.
  \end{equation}
  If there exists $\kappa >0$ such that 
  $ \displaystyle \liminf_{|x| \rightarrow \infty} \frac{V(x)}{|x|^\kappa} >0,$ then 
  \begin{equation}\label{eq:cbeta}
    c(\beta) =  
    \begin{cases}
      \mathcal O(\log\be), & \mathrm{as }\, \be \rightarrow 0,\\
      \mathcal O(\be), & \mathrm{as }\, \be \rightarrow \infty.
    \end{cases}
  \end{equation}
  If $V$ satisfies \eqref{eq:Vgrowth}, then the same holds true when replacing
  $\DBL$ by $\WAS_1.$
\end{theorem}

Theorem \ref{th:theococo} is proved in Section \ref{se:theococo}. The proof
relies on the Coulomb transport inequality of Theorem \ref{th:coultransp},
hence the assumption \eqref{eq:Vgrowth} in the second part. The constraint
\eqref{eq:weakLapGrowth} in Theorem \ref{th:theococo} that $\Delta V$ does not
grow faster than $V$ is technical. It allows potentials growing like
$\ABS{\cdot}^\kappa$ for any $\kappa>0$ or $\exp\ABS{\cdot}$, but not
$\exp(\ABS{\cdot}^2)$. As for the regularity condition, we assume for
convenience that $V$ is $\cC^2$ but much less is required, see Remark
\ref{rm:regularity}.

If $\beta>0$ is fixed, then under the last set of assumptions of Theorem
\ref{th:theococo}, there exist constants $u,v>0$ depending on $\beta$ and $V$
only such that, for any $N\geq 2$ and
\begin{equation}
  \label{eq:rcond}
  r\geq 
  \begin{cases}
    v \sqrt{\frac{\log N}{N}}& \mbox{ if } d=2,\\
    vN^{-1/d} & \mbox{ if } d\geq 3,
  \end{cases}
\end{equation}
we have
\begin{equation}
  \label{eq:UBMaidaMaureltype}
  \dP_{V,\be}^N \Big(\WAS_1(\hat\mu_N,\mu_V)\ge r \Big)\leq \e^{- uN^2r^2}.
\end{equation}
Such a bound has been obtained in \cite[Th.~5]{maida-maurel-segala} for the
empirical measure of one-dimensional log-gases.

\begin{remark}[Sharpness]\label{rk:sharpness}
  The concentration of measure bound \eqref{eq:UBMaidaMaureltype} is of
  optimal order with respect to $N$, as one can check using \eqref{eq:LDP} and
  the stochastic dominance
  \[
  \dP_{V,\be}^N\Big(\DBL(\hat\mu_N,\mu_V)\geq r\Big) %
  \leq\dP_{V,\be}^N\Big(\WAS_1(\hat\mu_N,\mu_V)\geq r\Big).
  \]
  The concentration inequalities provided by Theorem \ref{th:theococo} are much
  more precise than the large deviation upper bound in \eqref{eq:LDP} since it
  holds for every $N\geq 2$ and the dependence in $r,\beta$ is explicit.
\end{remark}

Combined with the Borel--Cantelli lemma, Theorem \ref{th:theococo} directly
yields the convergence of $\hat\mu_N$ to $\mu_V$ in the Wasserstein distance
even when one allows $\beta$ to depend on $N$, provided it does not go to zero
too fast, thanks to \eqref{eq:cbeta}.

\begin{corollary}[$\WAS_1$ convergence]\label{co:wcv} 
  Under the last set of assumptions of Theorem \ref{th:theococo}, there
  exists a constant $\beta_V>0$ such that the following holds: Given any
  sequence of positive real numbers $(\be_N)$ such that
  \[
  \beta_N\geq \beta_V\frac{\log N}N
  \] 
  for every $N$ large enough, then
  under $\dP_{V,\be_N}^N$,
  \[
  \lim_{N\to\infty}\WAS_1(\hat\mu_N,\mu_V)=0
  \]
  with probability one, in any joint probability space. 
\end{corollary}

Moreover, if one keeps $\beta>0$ fixed and lets $r\to 0$ with $N$, then
Theorem \ref{th:theococo} is precise enough to yield the convergence of
$\hat\mu_N$ towards $\mu_V$ at the mesoscopic scale, that is after zooming on
the particle system around any fixed $x_0\in\R^d$ at the scale $N^{-s}$ for
any $0\leq s<1/d$. More precisely, set $\tau_{x_0}^{N^s}(x):=N^s(x-x_0)$ and
let $\tau^{N^s}_{x_0}\mu$ be the the push-forward of $\mu\in\cP(\R^d)$ by the
map $\tau^{N^s}_{x_0}$, characterized by
\begin{equation}
  \int f(x)\,\tau^{N^s}_{x_0}\mu(\dd x)=\int f\big(N^s(x-x_0)\big)\mu(\dd x)
\end{equation}
for any Borel function $f:\R^d\to\R$.

\begin{corollary}[Mesoscopic convergence]\label{cor:lolo}
  Under the first set of assumptions of Theorem \ref{th:theococo}, for any
  $\beta>0$ satisfying \eqref{eq:gamma} there exist constants $C,c>0$ such
  that, for any $x_0\in\R^d$, any $s\geq 0$, and any $N\geq 2$, we have when
  $d=2$,
  \begin{equation}\label{eq:lolo2}
    \dP_{V,\be}^N %
    \Big(\DIST_\BL\big({\tau^{N^s}_{x_0}}\hat\mu_N,{\tau^{N^s}_{x_0}}\mu_V\big)
    \geq  CN^s\sqrt{\frac{\log N}{N}}\Big)\leq \e^{-cN\log N},
  \end{equation}
  and we have instead when $ d\geq 3$,
  \begin{equation}
    \label{eq:lolod}
    \dP_{V,\be}^N %
    \Big(\DIST_\BL\big({\tau^{N^s}_{x_0}}\hat\mu_N,{\tau^{N^s}_{x_0}}\mu_V\big)
    \geq  CN^{s-1/d}\Big)\leq \e^{-cN^{2-2/d}}.
  \end{equation}
  Under the last set of assumptions of Theorem \ref{th:theococo}, the same
  holds true after replacing $\DBL$ by $\WAS_1.$
\end{corollary}

The proof of the corollary requires a few lines which are provided at the end
of Section \ref{se:theococo}.

One may also look for a concentration inequality where the constants are
explicit in terms of $V$. We are able to derive such a statement when $\Delta
V$ is bounded above. Recall that the Boltzmann--Shannon entropy of
$\mu\in\cP(\dR^d)$ is defined by
\begin{equation}\label{eq:BSent}
  \ENTS(\mu):=-\int\frac{\dd\mu}{\dd x}\log\frac{\dd\mu}{\dd x}\,\dd x,
\end{equation}
with $\ENTS(\mu):=+\infty$ if $\mu$ is not absolutely continuous with respect
to the Lebesgue measure. Note that $\ENTS(\mu)$ takes its values in
$\dR\cup\{+\infty\}$ and is finite if and only if $\mu$ admits a density $f$
with respect to the Lebesgue measure such that $f\log f\in L^1(\dR^d)$. In
particular, if $V\in\cC^2(\dR^d)$ then $S(\mu_V)$ is indeed finite.

\begin{theorem}[Concentration for potentials with bounded Laplacian]%
  \label{th:theococo2}%
  Assume that $V$ is $\cC^2$ on $\R^d$ and that there exists $D>0$ with
  \begin{equation}
  \label{eq:strongLapGrowth}
  \sup_{y\in\R^d}\Delta V(y)\leq D.
  \end{equation}
  Then Theorem \ref{th:theococo} holds with the following constants: 
  \begin{align*} 
    a & = \frac{1}{8C^V_\BL},\\
    b & = \frac12\Big(\frac{1}{C_{\BL}^V}
        +\ENE(\la_1)
        +\frac{D}{2(d+2)}\Big),\\
    c(\be) & =\frac{\be}{2}\int V(x) \mu_V(\dd x)-\ENTS(\mu_V)
             +\log\int \e^{-\frac\be2(V(x)-\IND_{d=2}\log(1+|x|^2))}\dd x,
  \end{align*}
  where the constant $C_{\BL}^V$ is as in Theorem \ref{th:coultransp}, and
  where $ \ENE(\la_1)$ is the Coulomb energy of the uniform law $\la_1$ on the
  unit ball of $\dR^d$,
  \[
  \ENE(\la_1)=
  \begin{cases}
    \, \frac14 & \mbox{ if } d=2,\\
    \, \frac{2d}{d+2}  & \mbox{ if } d\geq 3.
  \end{cases}
  \]
  Under \eqref{eq:Vgrowth}, the constants $a$, $b$, $c(\beta)$ are given by
  the same formulas after replacing the constant $C_{\BL}^V$ by
  $C_{\WAS_1}^V.$
\end{theorem}

Theorem \ref{th:theococo2} is also proved in Section \ref{se:theococo}. 

\begin{remark}[Regularity assumptions]\label{rm:regularity}
  The assumption that $V$ is $\cC^2$ in theorems \ref{th:theococo} and
  \ref{th:theococo2} is made to ease the presentation of the results and may
  be considerably weakened. As we can check in the proofs, if we assume for
  instance that $V$ is finite on a set of positive Lebesgue measure and
  if the Boltzmann--Shannon
  entropy $\ENTS(\mu_V)$ of the equilibrium measure $\mu_V$ is finite, and if
  $V=\tilde V+h$ where $h:\R^d\to\R\cup\{+\infty\}$ is superharmonic and
  $\tilde V$ is twice differentiable such that $\Delta\tilde V$ satisfies the
  condition \eqref{eq:weakLapGrowth}, resp.\ \eqref{eq:strongLapGrowth}, then
  same conclusion as in Theorem \ref{th:theococo}, resp.\ Theorem
  \ref{th:theococo2}, holds for the $\DBL$ metric. If moreover $V$ satisfies
  the growth assumption \eqref{eq:Vgrowth}, then they also hold for $\WAS_1.$
\end{remark}

A remarkable consequence of Theorem \ref{th:theococo2} is an improvement of a  concentration
inequality in random matrix theory. Indeed, when $d=2$, $V=\ABS{\cdot}^2$, and
$\be=2$, the Coulomb gas
\[
\dd\dP_{\ABS{\cdot}^2,2}^N(x_1,\ldots,x_N)
=\frac{1}{\pi^{N^2}\prod_{k=1}^Nk!}
\prod_{i<j}\ABS{x_i-x_j}^2\e^{-N\sum_{i=1}^N\ABS{x_i}^2}
\dd x_1\cdots\dd x_N
\]
coincides with the joint law of the eigenvalues of the Ginibre ensemble of
$N\times N$ non-Hermitian random matrices, after identifying
$\dC^N\simeq(\dR^2)^N$, see \citep{MR0173726}. The associated equilibrium
$\mu_V$ measure is the circular law:
\[
\mu_\circ(\dd x):=\frac{\IND_{\{x\in\R^2:\;|x|\leq 1\}}}{\pi}\dd x.
\]
We also refer to \citep{MR2932638}, \citep[Ch.~15]{MR2129906,MR2641363},
\citep{MR2908617} for more information. The following result is a direct
application of Theorem \ref{th:theococo2} together with easy computations for
the constants. 

\begin{corollary}[Concentration for eigenvalues of Gaussian random matrices]%
  \label{co:ginibre}%
  Let $X$ and $Y$ be independent real Gaussian random variables with mean $0$
  and variance $1/2$. For any $N$, let $\mathbf{M}_N$ be the $N\times N$
  random matrix with i.i.d.\ entries distributed as the random complex
  variable $X+iY$. Let $\dP_N$ be the joint law of the eigenvalues
  $x_{1},\ldots,x_{N}$ of $\frac{1}{\sqrt{N}}\mathbf{M}_N$, and consider the
  empirical spectral measure
  \[
  \hat\mu_N:=\frac{1}{N}\sum_{i=1}^N\de_{x_{i}}.
  \]
   Then, for any $N\geq 2$ and any
  $r>0$, we have
  \begin{equation}\label{eq:concGinibre}
  \dP_N \Big(\WAS_1(\hat\mu_N,\mu_\circ)\ge r \Big)
  \le \e^{ - \frac{1}{4C} N^2 r^2 
    +\frac{1}{2} N\log N \mathbf
    + N [\frac{1}{C}
    +\frac32 -\log\pi]}\, ,
  \end{equation}
  where $C$ is the constant $C^{\WAS_1}_{\ABS{\cdot}^2}$ as in Theorem
  \ref{th:coultransp} for $d=2$.
\end{corollary}

For a fixed test function, similar concentration inequalities for the Ginibre ensemble have been obtained by  \cite{pemantle-peres} and \cite{breuer-duits},
but with asymptotic rate $\e^{- Nr}$ instead of $\e^{- N^2r^2}.$ Corollary \ref{co:ginibre} provides the almost sure convergence in the Wassertein
$\WAS_1$ metric of the empirical spectral measure to the circular law. We
refer to \citep{meckes-meckes} for a survey on empirical spectral measures,
focusing on coupling methods and Wasserstein distances, and covering in
particular the Ginibre ensemble. We do not know how to deduce the
concentration inequality provided by Corollary \ref{co:ginibre} from the
Gaussian nature of the entries of $\mathbf M_N$. Indeed, the eigenvalues of a
non-normal matrix are not Lipschitz functions of its entries, in contrast with
the singular values for which the Courant--Fischer formulas and the
Hoffman--Wielandt inequality hold, see \citep{MR2908617} for further details.

It is quite natural to ask about the behavior of the support of the random
empirical measure $\hat\mu_N$ under $\dP_{V,\be}^N$ when $N$ is large. The
following theorem gives an answer.

\begin{theorem}[Exponential tightness]\label{th:tightness}
  Assume $V$ is admissible, finite on a set of positive Lebesgue measure,
  satisfies the growth assumption
  $$
    \lim_{\ABS{x}\to\infty}\big(V(x) - (2+\varepsilon) \log|x|\, \IND_{d=2} \big)=+\infty,
  $$ 
  for some $\varepsilon>0$, and that $\mu_V$ has finite Boltzmann--Shannon
  entropy $\ENTS(\mu_V)$. Then, for any $\beta$ satisfying \eqref{eq:gamma},
  there exist constants $r_0>0$ and $c>0$ which may depend on $\be$ and $V$
  such that, for any $N$ and $r\ge r_0$, we have,
  \[
  \dP_{V,\be}^N(\mathrm{supp}(\hat\mu_N)\not\subset B_r) =
  \dP_{V,\be}^N\bigr(\max_{1\le i\le N}|x_i|\ge r\bigr)
  \le\e^{-cNV_*(r)},
  \]
  where $B_r:=\{x\in\dR^d:|x|\leq r\}$ and $V_*(r):=\min_{|x|\geq r}V(x)$.
\end{theorem}

The proof of Theorem \ref{th:tightness} is given in Section \ref{se:tightness}.

Since for any admissible $V$ we have $V_*(r)\to+\infty$ as $r\to\infty$, the
theorem states that the maximum modulus $\max |x_i|$ is exponentially tight at
speed $N$. In particular, the Borel--Cantelli lemma gives that
\[
\limsup_{N\to\infty}\max_{1\le i\le N}\ABS{x_i}<\infty
\]
holds with probability one in any joint probability space. 

Theorem \ref{th:tightness} allows to obtain $\WAS_p$ versions of Theorem
\ref{th:theococo}, without much efforts. Indeed, Theorem \ref{th:tightness}
allows to restrict to the event $\{\max_{1\leq i\leq N}\ABS{x_i}\leq M\}$ with
high probability. On the other hand, the distances $\DIST_\BL$ and $\WAS_p$
for any $p\geq1$ are all equivalent on a compact set: for any $p\geq1$ and any
probability measures $\mu,\nu$ supported in the ball of $\dR^d$ of radius
$M\geq1$,
\[
\WAS_p^p(\mu,\nu)
\leq (2M)^{p-1}\WAS_1(\mu,\nu)
\leq M(2M)^{p-1}\DIST_\BL(\mu,\nu).  
\]
For example, for $p=2$, this gives, under \eqref{eq:Vgrowth}, an inequality of
the form
\[
\dP_{V,\be}^N(W_2(\hat\mu_N,\mu_V)\ge r) \leq 2\e^{-cN^{3/2}r^2}.
\]
The same idea allows to deduce the almost sure convergence of the empirical
measure to the equilibrium measure with respect to $\WAS_p$ for any $p\geq1$.


\subsection{Notes, comments, and open problems}  

\subsubsection{Classical, free, and Coulomb transport inequalities}

The inequalities involving the Wasserstein $\WAS_1$ metric in Theorem
\ref{th:clé} and \ref{th:coultransp} can be seen as Coulomb analogues of
classical transport inequalities, hence their names. Recall that
\emph{Kullback-Leibler divergence} of $\mu$ with respect to $\nu$ is given by
\begin{equation}\label{eq:entrel}
  \ENTH(\mu\mid\nu) %
  := \int\frac{\dd\mu}{\dd\nu}\log\frac{\dd\mu}{\dd\nu}\,\dd\nu,
\end{equation}
with $\ENTH(\mu\mid\nu):=+\infty$ if $\mu$ is not absolutely continuous with
respect to $\nu$. This functional is also known as relative entropy or free
energy, see for instance \citep{MR3434251}. A probability measure
$\nu \in\cP(\dR^d)$ satisfies the \emph{transport inequality} of order $p$
when for some constant $C>0$ and every $\mu\in\cP(\dR^d)$,
\begin{equation*}\label{eq:TI}\tag{$\bT$}
  \WAS_p(\mu,\nu)^2 \le C \ENTH(\mu\mid\nu). 
\end{equation*}
A vast vast literature is devoted to the study of transport inequalities, see
the surveys \citep{gozlan-leonard} and \citep{MR1964483} for more details.
Following for instance \citep{gozlan-leonard}, we want to emphasize the deep
link between \eqref{eq:TI} and Sanov's large deviation principle for the
empirical measure of i.i.d random variables with law $\nu$, where the rate
function is given by $\mu \mapsto \ENTH(\mu\mid\nu)$ and the speed is $N.$
Similarly, as already mentioned, the empirical measure of a Coulomb gas in any
dimension $d \ge 2$ has been shown by \citet{chafai-gozlan-zitt} to satisfy a
large deviation principle, this time with rate function $\ENE_V-\ENE_V(\mu_V)$
and speed $N^2.$ Moreover, one can check that if the support of $\mu$ is
contained in the one of $\mu_V$, then
$\ENE_V(\mu)-\ENE_V(\mu_V)=\ENE(\mu-\mu_V)$, and the squared Coulomb metric
$\ENE(\mu-\nu)$ can also be seen as the counterpart of the Kullback-Leibler
divergence $\ENTH(\mu\mid\nu).$ Inequalities of the same flavor for
probability measures on the real line, linking the Wassertein $\WAS_1$ or
$\WAS_2$ metrics with the functionals $\ENE(\mu-\nu)$ or
$\ENE_V(\mu)-\ENE_V(\mu_V)$ with $d=2,$ have been previously obtained in the
context of free probability by \citet{biane-voiculescu, MR2093762,MR2535467,
  maida-maurel-segala, popescu}. In this setting, these functionals are
usually referred to as free entropies and are related to the large deviation
principle due to \citet{benarous-guionnet} for the one-dimensional log-gas
associated to unitarily invariant ensembles in random matrix theory.

There is also a deep connection between the inequality \eqref{eq:TI} and
concentration of measure for Lipschitz functions. In particular,
\citet{MR838213} and later \citet{bobkov-gotze} have shown that \eqref{eq:TI}
with $p=1$ is equivalent to sub-Gaussian concentration, while \citet{MR2573565}
has shown that \eqref{eq:TI} with $p=2$ is equivalent to dimension-free
sub-Gaussian concentration. However, the way we deduce the concentration of
measure for Coulomb gases from the Coulomb transport inequality is of
different nature, and is inspired from \citet{maida-maurel-segala}.

\citet{MR2535467} and \citet{popescu} provide $W_2$ free transport inequalities.
It is therefore natural to ask about a $W_2$ variant of Theorem
\ref{th:coultransp}. In the same spirit, one can study Coulomb versions of
related functional inequalities such as the logarithmic Sobolev inequalities
(not the Hardy-Littlewood-Sobolev inequalities). It is not likely however that
Theorem \ref{th:clé} has an extension to other $\WAS_p$ distances, as
\citet{popescu} showed this is not true in the free setting.


\subsubsection{Varying potentials and conditional gases}

Theorem \ref{th:theococo} still holds when $V$ depends mildly on $N$. This can
be useful for conditional Coulomb gases. Indeed, the conditional law
$\dP_{V,\be}^N(\cdot\mid x_N)$ is again a Coulomb gas
$\dP_{\tilde V,\be}^{N-1}$ with potential
\[
\tilde V:=\frac{N}{N-1}V+\frac{2}{N-1}g(x_N-\cdot),
\]
which is covered by our results since $g$ is superharmonic, see Remark
\ref{rm:regularity}.

\subsubsection{Weakly confining potentials}

In this work, we consider only admissible potentials and we do not allow the
weakly confining potentials considered in \citep{MR2926763}, in relation in
particular with the Cauchy ensemble of random matrices studied by Forrester
and Krishnapur, see \cite[Sec.~2.8]{MR2641363}. In this case the equilibrium
measure $\mu_V$ is no longer compactly supported. The derivation of
concentration of measure inequalities for such Coulomb gases is still open.

\subsubsection{Random matrices}

It is likely that the concentration inequality provided by Corollary
\ref{co:ginibre} remains valid beyond the Gaussian case. This is still open,
even when the entries of $\mathbf{M}_N$ are i.i.d.\ $\pm1$ with probability
$1/2$.

\subsubsection{Inverse temperature and crossover}

With the parameter $\be$ depending on $N$, Corollary \ref{co:wcv} states that
$\WAS_1(\hat\mu_N,\mu_V)\to 0$ as $N\to\infty$ provided that $\be\geq
\beta_V\log N/N$. Maybe this still holds as long as $N\be \to\infty$. However,
when $\be$ is of order $1/N$, we no longer expect concentration around $\mu_V$
but, having in mind for instance \citep{MR1678526} or
\citep{2012PhRvL.109i4102A}, it is likely that there is still concentration of
measure around a limiting distribution which is a crossover between $\mu_V$
and a Gaussian distribution.

\subsection*{Acknowledgements} We would like to thank Sylvia Serfaty for her useful comments on this work.   

\subsection*{Notations} 

We denote by $B_R=\{x\in\R^d:\ABS{x}\le R\}$ the centered closed ball of
radius $R$ of $\dR^d$, by $\la_R$ the uniform probability measure on $B_R$, by
$\VOL(B)$ the Lebesgue measure of a Borel set $B\subset\R^d$, and by
$\SUPP(\mu)$ the support of $\mu\in\cP(\dR^d)$.

\section{Proof of Theorem \ref{th:clé}}
\label{se:th:clé}

\subsection{Core of the proof}

We now give a proof of Theorem \ref{th:clé} up to technical lemmas. We
postpone their proofs to the next subsection.

The starting point is the Kantorovich--Rubinstein dual representation
\eqref{eq:KRduality} of the Wasserstein $\WAS_1$ metric, namely
\[
\WAS_1(\mu,\nu)=\sup_{\LIP{f}\le 1}\int f(x)(\mu-\nu)(\dd x).
\]
A theorem due to Rademacher states that if $\LIP{f}<\infty$ then $f$ is
differentiable almost everywhere and $\LIP{f}=\NRM{\nabla f}_\infty$, where
$\NRM{\cdot}_\infty$ stands for the $L^\infty$ norm.

The following lemma shows that we can localize the functions in the supremum,
provided the measures are supported in a compact set. 

\begin{lemma}[Localization]%
  \label{le:W1nice}%
  For any $D\subset\R^d$ compact, there exists a compact set $D_+\subset
  \R^d$ such that:
  \begin{itemize}
    \item[(a)]  $D\subset D_+$
    \item[(b)]  $\VOL(D_+)>0$
    \item[(c)] for every  $\mu,\nu\in\cP(\R^d)$ supported in $D$,
      \begin{equation}\label{eq:clé1}
        \WAS_1(\mu,\nu)
        =\sup_{\substack{f\in\cC(D_+)\\\LIP{f}\le 1}}\int f(x)(\mu-\nu)(\dd x),
      \end{equation}
      where $\cC(D_+)$ is the set of continuous functions $f:\dR^d\to\dR$
      supported in $D_+$.
    \end{itemize}
\end{lemma}


Next, let $\mu$ and $\nu$ be any probability measures supported in a compact
set $D\subset\R^d$ satisfying $\ENE(\mu)<\infty$ and $\ENE(\nu)<\infty,$ where
the Coulomb energy $\cE$ has been defined in \eqref{eq:CE}. Assume further for
now that $\mu-\nu$ has a $\cC^\infty$ density $h$ with respect to the Lebesgue
measure which is compactly supported; we will remove this extra assumption at
the end of this section.

Set $\eta:=\mu-\nu$ for convenience. We first gather a few properties of the
Coulomb potential of $\eta$ defined by
\[
U^\eta(x):=g*\eta(x)=\int g(x-y) h(y)\dd y,\qquad x\in\R^d.
\] 
Since $g\in L^1_{\mathrm{loc}}(\R^d)$, we see that $U^{\eta}$ is well defined.
Since moreover $h\in\cC^\infty(\R^d)$ so does $U^\eta$. We claim that $\nabla
U^\eta\in L^2(\R^d)$. Indeed, if we set $\alpha_d:=\max (d-2,1)$, then
\[
\nabla g(x)
= -\alpha_d \frac{x}{\ABS{x}^d}.
\]
Using that $\eta$ has compact support and that $\eta(\R^d)=0$, we have as
$|x|\to\infty$,
\[
\ABS{\nabla U^{\eta}(x)}^2 %
=\alpha^2_d %
\iint \frac{(x-y)\cdot(x-z)}{\ABS{x-y}^d\ABS{x-z}^d}\,\eta(\dd y)\eta(\dd z) %
=\frac{\alpha_d^2}{\ABS{x}^{2d-1}}(1+o(1)),
\]
from which our claim follows. Finally, Poisson's equation \eqref{eq:Poisson}
yields
\begin{equation}\label{eq:DeltaUeta}
  \De U^\eta(x)=-c_d \,h(x),\qquad x\in\R^d.
\end{equation}
Indeed, \eqref{eq:Poisson} states that for any test function $\varphi\in
\cC^\infty(\R^d)$ with compact support,
\[
\int\De\varphi(y) g(y)\dd y=-c_d\, \varphi(0),
\]
see \cite[Th.~6.20]{lieb-loss}, and \eqref{eq:DeltaUeta} follows
by taking $\varphi(y)= h(x-y)$.\\

Now, take any Lipschitz function $f\in\cC(D_+)$, where the compact set $D_+$
has been introduced in Lemma \ref{le:W1nice}. We have by \eqref{eq:DeltaUeta}
that
\begin{equation}
  \int f(x)(\mu-\nu)(\dd x) %
  =\int f(x)h(x)\dd x %
  =-\frac1{c_d}\int  f(x)\De U^{\eta}(x)\dd x.
\end{equation}
Because $\nabla U^\eta,\nabla f\in L^2(\R^d)$ and $\De U^\eta\in
C^\infty(\R^d)$ with compact support, one can use integration by parts (see
for instance \cite[Th.~7.7]{lieb-loss}) to obtain
\begin{equation}
  -\int  f(x)\De U^{\eta}(x)\dd x=\int \nabla f(x)\cdot \nabla U^{\eta}(x)\dd x.
\end{equation}
By using the Cauchy--Schwarz inequality in $\dR^d$ and then in $L^2(\R^d)$, we
have
\begin{align}
  \int \nabla f(x)\cdot \nabla U^{\eta}(x)\dd x
  &\le\int |\nabla f(x)| |\nabla U^{\eta}(x)|\dd x \nonumber\\
  &\le\LIP{f} \int_{D_+} |\nabla U^{\eta}(x)|\dd x\nonumber\\
  &\le\LIP{f}\Bigr(\VOL(D_+) \int|\nabla U^{\eta}(x)|^2\dd x\Bigr)^{1/2}.
\end{align}
By using again an integration by parts, 
\begin{equation}\label{eq:clé2}
  \int|\nabla U^{\eta}(x)|^2\dd x %
  =-\int U^{\eta}(x)\De U^{\eta}(x)\dd x %
  =c_d\int U^{\eta}(x)h(x)\dd x=c_d\, \ENE(\eta).
\end{equation}
Finally, by combining \eqref{eq:clé1}--\eqref{eq:clé2} and Lemma
\ref{le:W1nice} we get a proof of Theorem \ref{th:clé} under the extra
assumption that $\mu$ and $\nu$ have a $\cC^\infty$ density, with
$C_D:=\VOL(D_+)$.

We finally remove this extra assumption by using a density argument. Assume
$\mu$ and $\nu$ are probability measures supported in a compact set
$D\subset\R^d$ with $\ENE(\mu),\ENE(\nu)<\infty$. Let $\vphi:\dR^d\to\dR$ be
such that
\[
\vphi\in \cC^\infty(B_1),\quad 
\vphi\ge0,\quad
\int \vphi (x) \dd x=1,
\]
and set $\vphi_\veps(x)=\veps^{-d}\vphi(\veps^{-1}x)$. Then both
$\vphi_\veps*\mu$ and $\vphi_\veps*\nu$ have a $\cC^\infty$ density supported
on an $\veps$-neighborhood of $D$. Thus, by the previous step, if $D^\delta$
stands for the $\delta$-neighborhood of the set $D,$ then for every
$0<\veps<\delta <1$,
\begin{equation}\label{eq:approxineqdom}
  \WAS_1(\vphi_\veps*\mu,\vphi_\veps*\nu)^2 %
  \le \VOL(D^\delta_+)\ENE(\vphi_\veps*\mu-\vphi_\veps*\nu).
\end{equation}
The probability measures $\vphi_\veps*\mu$ and $\vphi_\veps*\nu$ converge
weakly to $\mu$ and $\nu$ respectively as $\veps\to 0$. Since all these
measures are moreover supported in the compact set $D^1$,
$\WAS_1(\vphi_\veps*\mu,\vphi_\veps*\nu)\to \WAS_1(\mu,\nu)$ as $\veps\to 0$.
Thus, letting $\delta\to0$ in \eqref{eq:approxineqdom}, it is enough to prove
that $\ENE(\vphi_\veps*\mu-\vphi_\veps*\nu)\to\ENE(\mu-\nu)$ as $\veps\to0$ in
order to complete the proof of the theorem with $C_D:= \VOL(D_+)$. This is a
consequence of the next lemma, which is proven at the end of the section and
may be of independent interest. Consider the bilinear form
\begin{equation}\label{eq:bilinear}
  \ENE(\mu,\nu)=\iint g(x-y)\mu(\dd x)\nu(\dd y),
\end{equation}
which is well defined for any $\mu,\nu\in\cP(\R^d)$ with compact support. 

Note that the proof above is not that far from inequalities between Sobolev
norms.

\begin{lemma}[Regularization]%
  \label{le:energysmoother}%
  Let $\mu,\nu$ be two positive finite Borel measures on $\dR^d$, with compact
  support, and satisfying $\ENE(\mu,\nu)<\infty$. Let $\vphi\in
  L^\infty(B_1)$ be positive, satisfying $\int \vphi (x)\dd x=1$, and set
  $\vphi_\veps(x)=\veps^{-d}\vphi(\veps^{-1}x)$. Then, we have
  \[
  \lim_{\veps\to0}\ENE(\vphi_\veps*\mu,\vphi_\veps*\nu)=\ENE(\mu,\nu).
  \]
\end{lemma}

\noindent Now, since $\mu$ and $\nu$ have the same total mass, compact
support, $\ENE(\mu),\ENE(\nu)<\infty$ and, as mentioned right after
\eqref{eq:CE}, $\ENE(\mu-\nu)\ge 0$, we have that $2\ENE(\mu,\nu)\le
\ENE(\mu)+\ENE(\nu)<\infty$. Thus Lemma \ref{le:energysmoother} applies and
gives that
\[
\ENE(\vphi_\veps*\mu-\vphi_\veps*\nu) %
= \ENE(\vphi_\veps*\mu)-2\ENE(\vphi_\veps*\mu,\vphi_\veps*\nu)+\ENE(\vphi_\veps*\nu)
\]
converges to $ \ENE(\mu)-2\ENE(\mu,\nu)+\ENE(\nu) =\ENE(\mu-\nu) $ as
$\veps\to 0$ and the proof of Theorem \ref{th:clé} is therefore complete up to
the technical lemmas \ref{le:W1nice} and \ref{le:energysmoother}.


\subsection{Proof of the technical lemmas \ref{le:W1nice} and
  \ref{le:energysmoother}}

\begin{proof}[Proof of Lemma \ref{le:W1nice}] Let $R >0$ be such that $D
  \subset B_R$. Given any function $f$ such that $\LIP{f}\le 1$ and $f(0)=0$,
  we claim one can find a function $\tilde f$ with support in $B_{4R}$ such
  that $\tilde f|_{D}=f|_{D}$ and $\LIP{\tilde f}\le 1$. Indeed, define the
  function $\tilde f$ as follows: for any $r \ge 0$ and $\te \in \mathbb
  S^{d-1}:=\{x\in\dR^d:\ABS{x}=1\}$,
  \[
  \tilde f(r\te) = \left\{
    \begin{array}{ll}
      f(r\te), & \textrm{ if } r \le R,\\
      f(R\te), & \textrm{ if } R \le r \le 2R,\\
      f(R\te) \frac{4R-r}{2R}, & \textrm{ if } 2R \le r \le 4R,\\
      0, & \textrm{ if } r \ge 4R.
    \end{array}
  \right.
  \]
  By construction the only non trivial point to check is $\LIP{\tilde f}\le
  1$, namely
  \begin{equation}
    \label{eq:LipToCheck}
    |\tilde f(r\te)-\tilde f(r^\prime\te^\prime)|\le |r\te-r^\prime\te^\prime|
  \end{equation}
  for every $r^\prime\ge r \ge 0$ and $\te, \te^\prime \in \mathbb
  S^{d-1}$. It is enough to prove \eqref{eq:LipToCheck} when both $r$ and $r'$
  belong to the same interval in the definition of $\tilde f$. But the latter
  is obvious except when $2R \le r \le r^\prime \le 4R$. In this case, we use
  that $|f(R\te)|\le R$, because $f(0)=0$ by assumption, in order to get
  \begin{align*}
    |\tilde f(r\te) -\tilde f(r^\prime \te^\prime)| 
    & \le |\tilde f(r\te) -\tilde f(r^\prime \te)| %
    + |\tilde f(r^\prime\te) -\tilde f(r^\prime \te^\prime)|\\
    & \le\frac{f(R\te)}{2R} %
    |r^\prime-r|+|f(R\te)- f(R \te^\prime)|\frac{4R-r^\prime}{2R}\\
    & \le  \frac{1}{2} |r^\prime - r| + R|\te- \te^\prime|\\ 
    & \le  \frac{1}{2} |r^\prime - r| + \frac{1}{2} r|\te- \te^\prime|\\ 
    & \le  |r\te - r^\prime \te^\prime|.
  \end{align*}
  We used in the last inequality that $|r\te - r\te^\prime| \le |r\te - r'
  \te^\prime|$, which holds true since $ r\te'$ is the orthogonal projection
  of $r'\te'$ onto $B_{ r}$ and $r\te\in B_{ r}$.

  As a consequence, for any $\mu,\nu$ probability measures supported in $D$,
  we obtain from \eqref{eq:KRduality},
  \begin{align*}
    \sup_{\substack{f\in\cC(B_{4R})\\\LIP{f}\le 1}}\int f(x)(\mu-\nu)(\dd x)
     \le \WAS_1(\mu,\nu) & = \sup_{\LIP{f}\le 1}\int f(x)(\mu-\nu)(\dd x)\\
    & =\sup_{\substack{\LIP{f}\le 1\\f(0)=0}}\int f(x)(\mu-\nu)(\dd x)\\
    & =\sup_{\substack{\LIP{f}\le 1\\f(0)=0}}\int \tilde f(x)(\mu-\nu)(\dd x)\\
    & \le\sup_{\substack{\tilde f\in\cC(B_{4R})\\\LIP{\tilde f}\le 1}}
    \int \tilde f(x)(\mu-\nu)(\dd x).
  \end{align*}
  Thus,
  \begin{equation}\label{eq:localize}
    \WAS_1(\mu,\nu)%
    =\sup_{\substack{f\in\cC(B_{4R})\\\LIP{f}\le 1}}\int f(x)(\mu-\nu)(\dd x),
  \end{equation}
  which proves the lemma with $D_+=B_{4R}$.
\end{proof}

Before going further, we show a simple but quite useful lemma, in the spirit
of Newton's theorem, see e.g.\ \cite[Th.~9.7]{lieb-loss}. Recall that $\la_R$
is the uniform probability measure on the ball $B_R$ in $\dR^d$.

\begin{lemma}[Superharmonicity]%
  \label{le:superharm}%
  For every $x\in\R^d$,
  \[
  \iint g(x+u-v)\la_R(\dd u)\la_R(\dd v)\le g(x).
  \]
\end{lemma}

\begin{proof} 
  Since $g$ is superharmonic, we have
  \[
  \int g(x+u)\la_R(\dd u)\le g(x),\qquad x\in\R^d.
  \]
  As a consequence, 
  \begin{align*}
    &\iint g(x+u-v)\la_R(\dd u)\la_R(\dd v)\\
    & =\int g(x+u)-
    \underbrace{\left(g(x+u)
        -\int g(x+u-v)\la_R(\dd v)\right)}_{\ge 0}\la_R(\dd u)\\
    & \le \int g(x+u)\la_R(\dd u)\le g(x).
  \end{align*}
\end{proof}

\begin{proof}[Proof of Lemma \ref{le:energysmoother}] We set for convenience
  \begin{eqnarray*}
  g_\veps(x-y)& := &\iint g(x-y+u-v)\vphi_\veps(u)\vphi_\veps(v)\dd u\dd v,\\
  & = & \iint g(x-y+\veps(u-v))\vphi(u)\vphi(v)\dd u\dd v,
  \end{eqnarray*}
  which is positive when $d\ge 3$. When $d=2$, by dilation we can assume
  without loss of generality that $\mu$ and $\nu$ are supported in $B_{1/4}$,
  so that both $g(x-y)$ and $g_\veps(x-y)$ are positive for every $(x,y)\in
  \SUPP(\mu)\times \SUPP(\nu) $ and every $\veps<1/4$.
  
  First, Fubini--Tonelli's theorem  yields
  \begin{align*}
    \ENE(\vphi_\veps*\mu,\vphi_\veps*\nu) %
    & =\iint %
    \left( 
      \iint g(u-v)\vphi_\veps(u-x)\vphi_\veps(v-y)\dd u\dd v\right)
    \mu(\dd x)\nu(\dd y)\\
    & =\iint g_\veps(x-y)\mu(\dd x)\nu(\dd y). 
  \end{align*}
  For every $x,y\in\R^d$ satisfying $x\neq y$, the map $(u,v)\mapsto
  g(x-y+\veps(u-v))$ is bounded on $B_1\times B_1$ for every $\veps$ small
  enough and converges pointwise to $g(x-y)$ as $\veps\to0.$ Moreover, since
  $\ENE(\mu,\nu)<\infty$ yields $\mu\otimes\nu(\{x=y\})=0$, by dominated
  convergence, $ g_\veps(x-y)$ converges $\mu\otimes\nu$-almost everywhere to
  $g(x-y)$. Next, we have
  \begin{align}
    \label{eq:gepsBound}
    g_\veps(x-y)
    & =\iint g(x-y+u-v)\vphi_\veps(u)\vphi_\veps(v)\dd u\dd v\nonumber\\
    & \le  \NRM{\vphi}_\infty^2 %
    \VOL(B_1)^2\iint g(x-y+u-v)\la_{\veps}(\dd u)\la_{\veps}(\dd v),
  \end{align}
  where $\la_{\veps}$ is the uniform probability measure on $B_\veps$. Thus,
  by Lemma \ref{le:superharm}, we obtain
  \[
  g_\veps(x-y)\le \NRM{\vphi}_\infty^2\VOL(B_1)^2\, g(x-y)
  \]
  and, since $g(x-y)$ is $\mu\otimes\nu$-integrable by assumption, the lemma
  follows by dominated convergence again.
\end{proof}

\section{Proof of Theorem \ref{th:coultransp}}
\label{se:th:coultransp}

We first prove Lemma \ref{co:localcoultransp}, which is a weak version of
Theorem \ref{th:coultransp} restricted to compactly supported
measures. This is an easy corollary of Theorem \ref{th:clé} and the
Euler--Lagrange equations which characterize the equilibrium measure. The
latter state that, if we consider the Coulomb potential of the equilibrium
measure $U^{\mu_V}=g*\mu_V$, then there exists a constant $e_V$ such that
\begin{equation}\label{eq:EL}
  2U^{\mu_V}(x)+V(x) 
  \begin{cases}
    =e_V & \mbox{ for q.e. } x\in \SUPP(\mu_V)\\
    \ge e_V & \mbox{ for q.e. } x\in\dR^d,
  \end{cases}
\end{equation}
where ``q.e.'' stands for quasi-everywhere and means up to a set of null
capacity. Conversely, $\mu_V$ is the unique probability measure satisfying
\eqref{eq:EL}. This characterization goes back at least to \citet{frostman} in
the case where the potential $V$ is constant on a bounded set and infinite
outside. For the general case, see \citep{saff-totik},
\citep{chafai-gozlan-zitt}, \cite[Th.~2.1]{serfaty-course}, and references
therein.

\begin{lemma}%
  \label{co:localcoultransp}%
  If $D\subset\R^d$ is compact then for any
  $\mu\in\cP(\R^d)$ supported in $D$,
  \[
  \DIST_{\BL}(\mu,\mu_V)^2 %
  \leq \WAS_1(\mu,\mu_V)^2 %
  \leq C_{D \cup\SUPP(\mu_V)} \big(\ENE_V(\mu)-\ENE_V(\mu_V)\big),
  \]
  with $C_{D \cup\SUPP(\mu_V)}$ as in Theorem \ref{th:clé}.
\end{lemma}

\begin{proof} 
  Since the statement is obvious when $\ENE(\mu)=+\infty$, one can assume
  $\ENE(\mu)<\infty$. By integrating \eqref{eq:EL} over $\mu_V$, we first
  obtain
  \begin{equation}\label{eq:CV}
    e_V = 2\,\ENE(\mu_V)+\int V(x)\mu_V(\dd x).
  \end{equation}
  Integrating again \eqref{eq:EL} but over $\mu$, we get together with
  \eqref{eq:CV},
  \[
  2\ENE(\mu,\mu_V)+\int V(x)\mu(\dd x)\ge 2\ENE(\mu_V)+\int V(x)\mu_V(\dd x),
  \]
  where we used the notation \eqref{eq:bilinear}.
  This inequality provides
  \begin{equation}\label{eq:eev}
    \ENE(\mu-\mu_V) %
    = \ENE(\mu)-2\ENE(\mu,\mu_V)+\ENE(\mu_V) %
    \le \ENE_V(\mu)-\ENE_V(\mu_V)
  \end{equation}
  and the result follows from Theorem \ref{th:clé} and \eqref{BL-W1 dominance}. 
\end{proof}

We are now in position to prove Theorem \ref{th:coultransp}. We follow closely
the strategy of the proof of \cite[Lem.~2]{maida-maurel-segala}.

\begin{proof}[Proof of Theorem \ref{th:coultransp}] 
  Let $\mu\in\cP(\R^d)$. We can assume without loss of generality that $\mu$
  has unbounded support and finite Coulomb energy. The strategy consists in
  constructing a measure $\tilde\mu$ with compact support, with smaller
  Coulomb energy, together with an explicit control on the distance of
  interest between $\mu$ and $\tilde\mu$, and then to use Lemma
  \ref{co:localcoultransp}. To do so, given any $R\ge 1$ consider the
  probability measures
  \[
  \mu_{\mathrm{core}}:=\frac{1}{1-\alpha}\mu|_{B_R},
  \qquad 
  \mu_{\mathrm{tail}}:=\frac{1}{\alpha}\mu|_{\R^d\setminus B_R}
  \qquad \text{where} \quad \alpha:=\mu(\R^d\setminus B_R),
  \]
  so that $\mu=(1-\alpha)\mu_{\mathrm{core}}+\alpha\mu_{\mathrm{tail}}$. We
  then define the probability measure
  \[
  \tilde\mu:=(1-\alpha)\mu_{\mathrm{core}}+\alpha\si,
  \]
  where $\si$ stands for the uniform probability measure on the unit sphere
  $\mathbb S^{d-1}$. In particular, $\tilde\mu$ is supported in $B_R$.
  
  First, we write
  \begin{multline}\label{eq:E-E-E}
    \ENE_V(\mu)-\ENE_V(\tilde\mu)
    =\alpha\Big\{\int V(x)(\mu_{\mathrm{tail}}-\si)(\dd x)
    +\alpha\ENE(\mu_{\mathrm{tail}})-\alpha\ENE(\si)\\
    + 2(1-\alpha)\ENE(\mu_{\mathrm{core}},\mu_{\mathrm{tail}}-\si)\Big\},
  \end{multline}
  where we used the notation \eqref{eq:bilinear}. 

  Note that, for any $x,y \in
  \R^d,$
  \begin{equation}\label{eq:lesurg}
    g(x-y)
    \ge -( \log(1+|x|)+ \log(1+|y|)\big)\IND_{d=2},
  \end{equation}
  so that
  \begin{equation}
    \ENE(\mu_{\mathrm{tail}})
    \ge -\int 2\log(1+|x|)\IND_{d=2} \,\mu_{\mathrm{tail}}(\dd x).
  \end{equation}
   Moreover, using alternatively,
    \begin{equation}\label{eq:lesurg}
    g(x-y)
    \ge -\log\big(2\max(1+|x|,1+|y|)\big)\IND_{d=2},
  \end{equation}
  we  also obtain
  \begin{equation}
    \ENE(\mu_{\mathrm{core}},\mu_{\mathrm{tail}})%
    \ge -\int \big(\log(1+|x|)+\log 2\big)\IND_{d=2}\, \mu_{\mathrm{tail}}(\dd x).
  \end{equation}
  Next, according to \cite[Lem.~1.6.1]{helms}, we have
  \[
  U^\si(x):= g*\si(x)=
  \begin{cases} 
    0 & \text{ if } |x|\le 1 \text{ and } d=2,\\
    1 & \text{ if } |x|\le 1 \text{ and } d\ge 3,\\
    g(x) & \text{ if } |x|> 1.
  \end{cases}
  \]
  In particular, $U^\si(x)\le 1$ for every $x\in\R^d$ and thus, using that
  $0\le\alpha\le 1$,
  \begin{equation}\label{eq:estimutildeB} 
    \alpha \ENE(\mu_{\mathrm{core}},\si) = \alpha \int U^{\si}(x)\mu_{\mathrm{core}}(\dd x) \le 1.
  \end{equation}
  By combining \eqref{eq:E-E-E}--\eqref{eq:estimutildeB},  we obtain
  \begin{equation}\label{eq:estimutildeC}
    \ENE_V(\mu)-\ENE_V(\tilde\mu)\\
    \ge  \alpha\int \big\{V(x) %
    - 2\log(1+|x|)\IND_{d=2}+C\big\}\mu_{\mathrm{tail}}(\dd x)
  \end{equation}
  for some $C>0$ which is independent on $\mu$.
  
  Next, using that
  $$
  \DIST_{\BL}(\mu,\tilde\mu) =\alpha\sup_{\substack{\LIP{f}\leq1\\\NRM{f}_\infty\leq1}}\int f(x)(\mu_{\mathrm{tail}}-\si)(\dd x)\leq 2\alpha,
  $$
  we obtain
  \begin{equation*}
    \ENE_V(\mu)-\ENE_V(\tilde\mu)-\DIST_{\BL}(\mu,\tilde\mu)^2 \\
    \ge  \alpha\int \big\{V(x) %
    - 2\log(1+|x|)\IND_{d=2}+C-4\big\}\mu_{\mathrm{tail}}(\dd x).
  \end{equation*}
  Since $V$ is admissible, the map
  \[
  x\mapsto V(x)- 2\log(1+|x|)\IND_{d=2}+C-4
  \]
  tends to $+\infty$ as $|x|$ becomes large, it follows that there exists
  $R>0$ large enough so that this map is positive on $\R^d\setminus B_R$.
  
  Since by
  construction $\mu_{\mathrm{tail}}$ is supported on $\R^d\setminus B_R$, this yields $$\ENE_V(\mu)-\ENE_V(\tilde\mu)-\DIST_{\BL}(\mu,\tilde\mu)^2 
    \ge 0,$$
  and in particular, 
  \[
    \DIST_{\BL}(\mu,\tilde\mu)^2%
    \le \ENE_V(\mu)-\ENE_V(\mu_V)%
    \quad\text{and}\quad%
    \ENE_V(\tilde\mu)\le \ENE_V(\mu).
  \]
  This yields inequality \eqref{eq:CoulombBL} since, by the triangle inequality and Lemma
  \ref{co:localcoultransp}, we have
  \begin{align}
  \label{eq:triangle ineq}
    \DIST_{\BL}(\mu,\mu_V)^2& \le (\DIST_{\BL}(\mu,\tilde\mu)+\DIST_{\BL}(\tilde\mu,\mu_V))^2\nonumber\\
    & \le 2\DIST_{\BL}(\mu,\tilde\mu)^2+2\DIST_{\BL}(\tilde\mu,\mu_V)^2\nonumber\\
    & \le 2\big(\ENE_V(\mu)-\ENE_V(\mu_V)\big) %
    +2C_{B_R \cup\SUPP(\mu_V)}\big(\ENE_V(\tilde\mu)-\ENE_V(\mu_V)\big)\nonumber\\
    & \le \underbrace{ 
      2\left(1
        +C_{B_R \cup\SUPP(\mu_V)}\right)}_{\ds C_{\BL}^V}\big(\ENE_V(\mu)-\ENE_V(\mu_V)\big).
  \end{align}
  Note that since by construction $R$ is independent on $\mu$, so does $C_{\BL}^V$.\\
  
  We now turn to the proof of inequality \eqref{eq:CoulombW1}. 
Under \eqref{eq:Vgrowth} we can assume
  that $\mu$ has a finite second moment
  \[
  \int |x|^2\mu(\dd x)<\infty
  \]
  and in particular that $\WAS_1(\mu,\mu_V)<\infty$. Indeed, because of the
  growth condition \eqref{eq:Vgrowth}, we have otherwise
  \[
  \int V(x)\mu(\dd x)=+\infty
  \]
  so that the right hand side of the inequality to prove is infinite.  
  Now, let $\ga>0$ be such that
  \[
  \liminf_{|x|\to\infty}\frac{V(x)}{4\ga|x|^2}>1,
  \]
  whose existence is ensured by  \eqref{eq:Vgrowth}. Using the representation
  \eqref{eq:KRduality}, we have the upper bound,
  \begin{align}\label{eq:estimutildeA}
    \WAS_1(\mu,\tilde\mu)& =\sup_{\LIP{f}\le 1}\int f(x)(\mu-\tilde\mu)(\dd x)\nonumber\\
    & =\alpha\sup_{\LIP{f}\le 1}\int f(x)(\mu_{\mathrm{tail}}-\si)(\dd x)\nonumber\\
    & =\alpha\sup_{\substack{\LIP{f}\le 1\\f(0)=0}}
    \int f(x)(\mu_{\mathrm{tail}}-\si)(\dd x)\nonumber\\
    & \le \alpha\int |x|(\mu_{\mathrm{tail}}+\si)(\dd x)\nonumber\\
    & \le 2\alpha\int |x|\mu_{\mathrm{tail}}(\dd x)\nonumber\\
    & \le 2\alpha\left(\int |x|^2\mu_{\mathrm{tail}}(\dd x)\right)^{1/2}.
  \end{align}
  Combined together with \eqref{eq:estimutildeC} this gives
   \begin{equation*}
    \ENE_V(\mu)-\ENE_V(\tilde\mu)-\gamma\WAS_1(\mu,\tilde\mu)^2 \\
    \ge  \alpha\int \big\{V(x) -4\ga|x|^2%
    - 2\log(1+|x|)\IND_{d=2}+C\big\}\mu_{\mathrm{tail}}(\dd x).
  \end{equation*}
  Therefore, using the definition of
  $\ga$ and  \eqref{eq:Vgrowth}, this yields as before that there exists $R\ge 1$ independent on $\mu$ such that
  \[
    \ga \WAS_1(\mu,\tilde\mu)^2%
    \le \ENE_V(\mu)-\ENE_V(\mu_V)%
    \quad\text{and}\quad%
    \ENE_V(\tilde\mu)\le \ENE_V(\mu).
  \]
 As in \eqref{eq:triangle ineq}, this gives the upper bound,
  \[
    \WAS_1(\mu,\mu_V)^2\le \underbrace{ 
      2\left(\frac1{\ga}
        +C_{B_R \cup\SUPP(\mu_V)}\right)}_{\ds C^V_{\WAS_1}}\big(\ENE_V(\mu)-\ENE_V(\mu_V)\big),
  \]
  and the proof of the theorem is complete. 
   \end{proof}

 \section{Proofs of Theorem \ref{th:theococo} and Theorem \ref{th:theococo2} }\label{se:theococo}

 We start with the following lower bound on $Z^N_{V,\be}$. Recall the
 definition \eqref{eq:BSent} of the Boltzmann--Shannon entropy $\ENTS$.

 \begin{lemma}[Normalizing constant]\label{le:partifunc} Assume that $V$ is
   admissible, finite on set of positive Lebesgue measure, and that
   $\ENTS(\mu_V)$ is
   finite. Then for any $N\ge 2$,
   \[ 
   Z_{V,\be}^N \ge  
   \exp\left\{-N^2\frac{\be}{2} \ENE_V(\mu_V)
     +N \left(\frac{\be}{2}\ENE(\mu_V)  +  \ENTS(\mu_V)\right)\right\}.
   \]
 \end{lemma}

Note that it is compatible with the convergence
$\lim_{N\to\infty}\frac{1}{N^2}\log Z_{V,\be}^N=-\frac{\be}{2}\ENE_V(\mu_V)$
from \citep{chafai-gozlan-zitt} and with the more refined asymptotic expansion
from \citep[Cor.~1.5]{leble-serfaty}.

\begin{proof}[Proof of Lemma \ref{le:partifunc}]
  We start by writing the energy \eqref{eq:HN} as 
  \begin{equation}\label{eq:HNeneq}
    H_N(x_1,\ldots,x_N)
    =\sum_{i\neq j} %
    \left(g(x_i-x_j)+\frac{1}2V(x_i)+\frac12V(x_j)\right)+\sum_{i=1}^NV(x_i).
  \end{equation}
  Since the entropy $\ENTS(\mu_V)$ is finite by assumption, the equilibrium
  measure $\mu_V$ has a density $\rho_V$. If we define the event
  $E_V^N:=\{x\in(\dR^d)^N:\prod_{i=1}^N\rho_V(x_i)>0\}$, then
  \begin{align*}
   & \log Z_{V,\be}^N\\
    &=\log\int\e^{-\frac{\be}{2} %
      \sum_{i\neq j}\left(g(x_i-x_j)+\frac{1}2V(x_i)+\frac12V(x_j)\right)}
    \prod_{i=1}^N\e^{-\frac{\be}{2}V(x_i)}\dd x_i\\
    &\ge\log
    \int_{E_V^N}\e^{-\frac{\be}{2} %
      \sum_{i\neq j}\left(g(x_i-x_j)+\frac{1}2V(x_i)+\frac12V(x_j)\right)}
    \prod_{i=1}^N\e^{-\frac{\be}{2}V(x_i)}\dd x_i\\
    &=\log\int_{E_V^N}%
    \e^{-\frac{\be}{2}\sum_{i\neq j} %
      \left(g(x_i-x_j)+\frac{1}2V(x_i)+\frac12V(x_j)\right)
      -\sum_{i=1}^N \big(\frac{\be}{2}V(x_i)+\log\rho_V(x_i)\big)}
    \prod_{i=1}^N\mu_V(\dd x_i)\\
    &\overset{(*)}{\ge} 
    -\frac{\be}{2}\int_{E_V^N}\sum_{i\neq j}%
    \left(g(x_i-x_j)+\frac{1}2V(x_i)+\frac12V(x_j)\right) %
    \prod_{i=1}^N\mu_V(\dd x_i)\\
    & \qquad \qquad -N\int \left(\frac{\be}{2} V(x)+\log\rho_V(x)\right) %
    \mu_V(\dd x)\\
    &=-N(N-1)\frac{\be}{2}\ENE_V(\mu_V)
    -N \frac{\be}{2} \int V(x) \mu_V(\dd x) +N \ENTS(\mu_V)\\
    & = - N^2 \frac{\be}{2}\ENE_V(\mu_V) %
    + N \frac{\be}{2}\ENE(\mu_V) %
    + N \ENTS(\mu_V),
  \end{align*}
  where $(*)$ comes from Jensen's inequality.\\
\end{proof}


Next, we use that the energy \eqref{eq:HN} is a function of the empirical measure
$\hat\mu_N$, and more precisely that one can express $\dP_{V,\be}^N$ in terms of the partition
function $Z_{V,\be}^N$ and the quantity
\[
\ENE_V^{\neq}(\hat\mu_N)-\ENE_V(\mu_V),
\]
where $\ENE_V^{\neq}(\hat\mu_N)$ is as in \eqref{eq:enevneq}. The idea is to
replace it by its continuous version $\ENE_V(\hat\mu_N)-\ENE_V(\mu_V)$, so as
to compare it to $\DBL(\hat\mu_N,\mu_V)$ or $\WAS_1(\hat\mu_N,\mu_V)$ thanks
to the Coulomb transport type inequalities of Theorem \ref{th:coultransp}.
However, this is meaningless because $\ENE_V(\hat\mu_N)$ is infinite since
$\hat\mu_N$ is atomic, but one can circumvent this problem by approximating
$\hat\mu_N$ with a regular measure $\hat\mu_N^{(\veps)}$, as explained in the
next lemma. The same strategy has been used by \cite{maida-maurel-segala}.
This regularization is also at the heart of the definition of the renormalized
energy of \cite{MR3455593}. See also \cite[Prop.~2.8 item
3]{chafai-gozlan-zitt}.

\begin{lemma}[Regularization]\label{le:regularization}
  For any $\veps>0$, let us define 
  \[
  \hat\mu_N^{(\veps)}:=\hat\mu_N*\la_\veps,
  \]
  where we recall that $\la_\veps$ is the uniform probability measure on the
  ball $B_\veps$. Then,
  \begin{equation}\label{eq:VLB}
    H_N(x_1,\ldots,x_N)%
    \ge  N^2\ENE_V(\hat\mu_N^{(\veps)})%
    -N\ENE(\la_\veps)%
    -N\sum_{i=1}^N (V*\la_\veps -V)(x_i).
  \end{equation}
\end{lemma}

Note the latter inequality is sharp. Indeed, since $g$ is harmonic on
$\R^d\setminus\{0\}$, the inequality \eqref{eq:ineqSuper} below is in fact an
equality as soon as $|x_i-x_j|\geq \veps$ for every $i\neq j$, and so does
\eqref{eq:VLB}.

\begin{proof} 
  The superharmonicity of $g$ yields, via
  Lemma \ref{le:superharm},
  \begin{align}\label{eq:ineqSuper}
    \sum_{i\neq j} g(x_i-x_j) %
    &\ge \sum_{i\neq j}\iint g(x_i-x_j+u-v)\la_\veps(\dd u)\la_\veps(\dd v)\\
    &=\sum_{i\neq j}\ENE(\de_{x_i}*\la_\veps,\de_{x_j}*\la_\veps)\nonumber\\
    &= N^2\ENE(\hat\mu_N^{(\veps)}) - N \ENE(\la_\veps)\nonumber\\
    &= N^2\ENE_V(\hat\mu_N^{(\veps)})-N\ENE(\la_\veps) -N\sum_{i=1}^NV*\la_\veps(x_i)
    \nonumber,
  \end{align}
  and \eqref{eq:VLB} follows. 
\end{proof}

We follow this common path to obtain both
theorems \ref{th:theococo} and \ref{th:theococo2} up to dealing with the
remaining terms due to the regularization. Indeed, one has to proceed differently depending on the
strength of the assumption we make on the Laplacian's potential $\Delta V$.
The proofs we provide use milder assumptions than in their
statements, as explained in Remark \ref{rm:regularity}.


\begin{proof}[Proofs of Theorem \ref{th:theococo} and Theorem
  \ref{th:theococo2}] Assume for now that $V$ is admissible (to ensure the
  existence of $\mu_V$), and finite on a set of positive Lebesgue measure (to
  ensure that $Z_{V,\beta}^N >0$), and that the equilibrium measure $\mu_V$
  has finite Boltzmann--Shannon entropy $\ENTS(\mu_V)$ (to use Lemma
  \ref{le:partifunc}). Let us fix $N\geq 2$, $r>0$ and $\be >0$ satisfying
  \eqref{eq:gamma} (to ensure that $Z_{V,\beta}^N <\infty$).

  \medskip
  
  \noindent\textit{Step 1: Preparation.} We start by putting aside a small
  fraction of
  the energy to ensure integrability in the next steps. Namely, for any
  $\eta\in(0,1)$ and any Borel set $A\subset(\R^d)^N$, we write
  \begin{align*}
    \dP_{V,\be}^N(A) 
    & = \frac{1}{Z_{V,\be}^N} %
    \int_A\e^{-\frac{\be}{2} H_N(x_1, \ldots, x_N)}\prod_{i=1}^N\dd x_i\\
    & = \frac{1}{Z_{V,\be}^N} %
    \int_A\e^{-\frac{\be}{2}(1-\eta) H_N(x_1, \ldots, x_N)}
    \e^{-\frac{\be}{2} \eta H_N(x_1, \ldots, x_N)} %
    \prod_{i=1}^N\dd x_i.
  \end{align*}
  Next, if we set
  \begin{equation}
    \label{eq:defC1}
    C_1:=-\frac{\be}{2}\ENE(\mu_V)  -  \ENTS(\mu_V),
  \end{equation}
  then we deduce from Lemma \ref{le:regularization} and then Lemma
  \ref{le:partifunc} that
  \begin{align}
    \label{eq:step1}
    \dP_{V,\be}^N(A) 
    & \le \frac{1}{Z_{V,\be}^N} \nonumber%
    \int_A\e^{-\frac{\be}{2} N(1-\eta) \left(N\ENE_V(\hat\mu_N^{(\veps)}) %
        -\ENE(\la_\veps) -\sum_{i=1}^N (V*\la_\veps-V)(x_i)\right)}\\
    &\qquad\qquad\times\e^{-\frac{\be}{2}\eta H_N(x_1, \ldots, x_N)}\prod_{i=1}^N\dd x_i\\
    & \le \e^{NC_1 +  \frac{\be}{2} N\ENE(\la_\veps)+ \frac\be 2N^2\eta\ENE_V(\mu_V)} 
    \e^{-\frac{\be}{2} N^2(1-\eta)\inf_A(\ENE_V(\hat\mu_N^{(\veps)})-\ENE_V(\mu_V))}\\
    &\qquad\qquad\times 
    \int\e^{\frac\beta2
      \Big(N(1-\eta)\sum_{i=1}^N (V*\la_\veps-V)(x_i)
      -\eta H_N(x_1, \ldots, x_N)\Big)} 
    \prod_{i=1}^N\dd x_i.\nonumber
  \end{align}
  In the second inequality we also used that $\cE(\lambda_\veps)>0$ as soon as
  $0<\veps<1$. Indeed, by a change of variables we have
  \begin{equation}\label{eq:E1expl}
    \ENE(\la_\veps)=
    \iint g(x-y)\la_\veps(\dd x)\la_{\veps}(\dd y)
    =
    \begin{cases}
      g(\veps)+\ENE(\la_1) & \text{ if $d=2$},\\
      g(\veps)\, \ENE(\la_1) & \text{ if  $d\ge 3$},
    \end{cases}
  \end{equation}
  and one can compute explicitly  that 
  \begin{equation}\label{eq:E1expl2}
    \ENE(\la_1)=
    \begin{cases}
      \frac14 & \mbox{ if } d=2,\\
      \frac{2d}{d+2}  & \mbox{ if } d\geq 3.
    \end{cases}
  \end{equation}

  \medskip
  
  \noindent\textit{Step 2: The remaining integral.} We now provide an upper
  bound on the integral
  \[
  \int
  \e^{\frac\beta2\Big(N(1-{\eta})\sum_{i=1}^N (V*\la_\veps-V)(x_i)-\eta H_N(x_1, \ldots, x_N)\Big)} 
  \prod_{i=1}^N\dd x_i.
  \]
  Having in mind Remark \ref{rm:regularity}, notice that if $V=\tilde V+h$
  with $h$ superharmonic then we have, for every $x\in\R^d$,
  \[
  (V*\la_\veps-V)(x)\leq (\tilde V*\la_\veps-\tilde V)(x).
  \]
  Thus one can assume without loss of generality that $h=0$. From now, assume
  further that $\tilde V=V$ is twice differentiable. Using that
  \[
  (V*\la_\veps-V)(x)=\int \big(V(x-\veps u)-V(x)\big)\la_1(\dd u),
  \]
  and noticing that  by symmetry, 
  \[
  \int  \ANG{\nabla V(x), u}\la_1(\dd u)=0,
  \] 
  we obtain by a Taylor expansion that, for any $x \in \dR^d,$
  \begin{equation}\label{eq:HessB1}
    (V*\la_\veps-V)(x) 
    \leq \frac{\veps^2}{2}\sup_{\substack{y\in\dR^d\\|y-x| \le \veps}}
    \int\ANG{\HESS(V)(y)u,u}\la_1(\dd u),
  \end{equation}
  where $\HESS(V)(y)$ is the Hessian matrix of the function $V$ at point $y$.
  Since the covariance matrix of $\la_1$ is a multiple of the identity,
  \begin{align*}
    \int \ANG{\HESS(V)(y)u, u}\la_1(\dd u) 
    & = \sum_{i,j=1}^d\HESS(V)(y)_{i,j}\int u_iu_j\,\la_1(\dd u)\\ 
    &= \TR\HESS(V)(y)\frac{1}{d}\int |u|^2\,\la_1(\dd u)\\
    &= \TR\HESS(V)(y)\int_0^1 r^{d+1}\,\dd r\\
    & = \frac{1}{d+2}\De V(y).
  \end{align*}
  Therefore, with   $(\Delta V)_+(y):=\max(\Delta V(y),0)$, we have
  \begin{equation}\label{eq:UBDV}
    (V*\la_\veps-V)(x) 
    \leq 
    \frac{\veps^2}{2(d+2)}\sup_{\substack{y\in\dR^d\\|y-x| \le \veps}}(\Delta V)_+(y).
  \end{equation}
  
  Next, we use that $g(x-y)\geq -\frac12\log(1+|x|^2)-\frac12\log(1+|y|^2)$
  when $d=2$, and that $g\geq 0$ when $d\geq 3$, in order to get
  \[
  H_N(x_1, \ldots, x_N)\geq N\sum_{i=1}^N\big(V(x_i)
  - \IND_{d=2}\log(1+|x_i|^2)\big).
  \]
  As a consequence, we obtain with \eqref{eq:UBDV} that, for any
  $0<\epsilon<1$,
  \begin{multline}\label{eq:boundRemainInt}
    \int\e^{\frac\beta2\Big(N(1-{\eta})\sum_{i=1}^N(V*\la_\veps-V)(x_i)-\eta
      H_N(x_1,\ldots,x_N)\Big)}
    \prod_{i=1}^N\dd x_i\\
    \leq \left(\int \e^{-\frac\beta2 N\Big(\eta (V(x)- \boldsymbol
        1_{d=2}\log(1+|x|^2)) -\veps^2\frac{1}{2(d+2)}\sup_{|y-x| < 1}(\De
        V)_+(y)\Big)} \dd x\right)^N.
  \end{multline}
  
  In the next steps, we will use the Coulomb transport inequalities to bound
  the term $\inf_A(\ENE_V(\hat\mu_N^{(\veps)})-\ENE_V(\mu_V))$ for sets $A$ of
  interest depending on which growth assumption we make on $V$. We will
  moreover bound the remaining integral \eqref{eq:boundRemainInt} using extra
  assumptions on $\Delta V$. In the next step, we assume $\Delta V$ is bounded
  above and provide a proof for Theorem \ref{th:theococo2}. In the final step,
  we only assume $\Delta V$
  grows not faster than $V$ and complete the proof of Theorem~\ref{th:theococo}.

  \medskip

  \noindent\textit{Step 3: Proof of Theorem \ref{th:theococo2}.} Assume in this
  step there exists $D>0$ such that
  \begin{equation}\label{eq:LapsupBnd}
    \sup_{x\in \dR^d}\Delta V(x)\leq D.
  \end{equation}
  We then set 
  \begin{equation}\label{eq:epseta1}
    \eta:=N^{-1},\qquad \veps:=N^{-1/d}.
  \end{equation} 
  and 
  \[
  C_2:=\log\int \e^{-\frac\be2(V(x)-\IND_{d=2}\log(1+|x|^2))}\dd x,
  \]
  which is finite by assumption \eqref{eq:gamma}. We obtain from
  \eqref{eq:boundRemainInt},
  \begin{equation}\label{eq:RemBndTh2}
    \int\e^{\frac\beta2
      \Big(N(1-{\eta})\sum_{i=1}^N(V*\la_\veps-V)(x_i)
      -\eta H_N(x_1,\ldots,x_N)\Big)} 
    \prod_{i=1}^N\dd x_i \leq \e^{NC_2+\frac\beta2N^{2-2/d} \frac D{2(d+2)}}.
  \end{equation}
  
  If we set 
  \begin{equation}\label{eq:ABL}
    A:=\Big\{\DIST_{\BL}(\hat\mu_N,\mu_V)\geq r\Big\},
  \end{equation}
  then by applying Theorem \ref{th:coultransp} to $\hat\mu_N^{(\veps)}$ we
  obtain,
  \begin{equation}\label{eq:TransportA}
    \inf_A\big(\ENE_V(\hat\mu_N^{(\veps)})-\ENE_V(\mu_V)\big)
    \ge \frac{1}{C^V_{\BL}} \inf_A\DIST_\BL(\hat \mu_N^{(\veps)},\mu_V)^2 
    \ge \frac{1}{C^V_{\BL}}\Big(\frac{r^2}2-  \veps^2\Big),
  \end{equation}
  where we used for the last inequality that
  \begin{align*}
    \frac12\DIST_\BL(\hat \mu_N,\mu_V)^2 
    & \leq \frac12\big(\DIST_\BL(\hat \mu_N^{(\veps)},\mu_V) %
    +\DIST_\BL(\hat \mu_N^{(\veps)},\hat\mu_N)\big)^2\\
    & \leq  \DIST_\BL(\hat \mu_N^{(\veps)},\mu_V)^2 %
    +\DIST_\BL(\hat \mu_N^{(\veps)},\hat\mu_N)^2\\
    & \leq \DIST_\BL(\hat \mu_N^{(\veps)},\mu_V)^2+\veps^2.
  \end{align*} 
  Setting $c(\be):=C_1+C_2+\frac\be 2 \ENE_V(\mu_V)$ and noticing it is the
  same constant $c(\be)$ as in Theorem \ref{th:theococo2}, we obtain by
  combining \eqref{eq:step1}, \eqref{eq:E1expl}--\eqref{eq:E1expl2},
  \eqref{eq:RemBndTh2}, and \eqref{eq:TransportA},
  \begin{equation}\label{eq:FINALineq}
    \dP_{V,\be}^N(A) %
    \le \e^{-\frac{\be}{4C^V_\BL} N(N-1)r^2  
      +\mathbf 1_{d=2}( \frac{\be}{4} N \log N) 
      +\frac\beta2N^{2-2/d} 
      [\frac1{C^V_\BL}+\frac D{2(d+2)}+\cE(\lambda_1)]+Nc(\be)}.
  \end{equation}
  Since $N-1\geq N/2$ for every $N\geq 2$, this gives the desired constants
  $a$ and $b$.
  
  If one assumes the stronger growth assumption \eqref{eq:Vgrowth}, then the
  exact same line of arguments with $A:=\{\WAS_1(\hat\mu_N,\mu_V)\geq r\}$
  shows inequality \eqref{eq:FINALineq} holds true after replacing the
  constant $C^V_\BL$ by $C^V_{\WAS_1}$, which completes the proof of
  Theorem~\ref{th:theococo2}.

  \medskip
  
  \noindent\textit{Step 5: Proof of Theorem \ref{th:theococo}.} In this last
  part, we assume that
  \begin{equation}\label{eq:LapGrowthbis}
    \limsup_{|x|\to\infty}\Bigg(\frac1{V(x)} 
    \displaystyle\sup_{\substack{y\in\dR^d\\|y-x| < 1}}\Delta V(y)\Bigg)< 2(d+2)
  \end{equation}
  in place of \eqref{eq:LapsupBnd}, and  set instead,
  \begin{equation}\label{eq:epseta2}
    \eta:=N^{-2/d}, \qquad \veps:=N^{-1/d},
  \end{equation}
  so that \eqref{eq:boundRemainInt} now reads 
  \begin{multline}\label{eq:step5}
    \int\e^{\frac\beta2
      \Big(N(1-{\eta})\sum_{i=1}^N (V*\la_\veps-V)(x_i)
      -\eta H_N(x_1,\ldots,x_N)\Big)}
    \prod_{i=1}^N\dd x_i\\
    \leq \left(\int \e^{-\frac\beta2 N^{1-2/d}\Big( V(x)- 
        \IND_{d=2}\log(1+|x|^2)
        -\frac{1}{2(d+2)}\sup_{|y-x| < 1}(\De V)_+(y)\Big)} \dd x\right)^N.
  \end{multline}
  Combining $V$ being admissible and \eqref{eq:LapGrowthbis} yields constants
  $u>0$ and $v\in\R$ such that, for every $x\in\R^d$,
  \[
  V(x)- \IND_{d=2}\log(1+|x|^2)-\frac{1}{2(d+2)}\sup_{|y-x|<1}(\De V)_+(y)%
  \geq 2uV(x)-v.
  \]
  This provides with \eqref{eq:step5},
  \begin{align*}
    \int\e^{\frac\beta2\Big(N(1-{\eta}) \sum_{i=1}^N(V*\la_\veps-V)(x_i)
      -\eta H_N(x_1,\ldots,x_N)\Big)} \prod_{i=1}^N\dd x_i
    & \leq \e^{\frac\beta2 vN^{2-2/d}} \left(\int \e^{-\beta uV(x)} \dd x\right)^N \\
    &\leq \e^{\frac\beta2 vN^{2-2/d} + Nw(\beta)},
  \end{align*}
  where we set 
  \begin{equation}\label{eq:cvbe}
    w(\beta) := \log \int \e^{-\beta uV(x)} \dd x.
  \end{equation}
  Taking $A:=\{\DIST_\BL(\hat\mu_N,\mu_V)\geq r\}$, this gives with
  \eqref{eq:step1}, \eqref{eq:E1expl}, \eqref{eq:TransportA},
  \eqref{eq:epseta2},
  \[
  \dP_{V,\be}^N(A) 
  \le \e^{-\frac{\be N^2 r^2}{4C^V_\BL} (1-N^{-\frac2d})
    +\mathbf 1_{d=2}( \frac{\be}{4} N \log N) 
    +\frac\be2N^{2-\frac2d}\big[\ENE_V(\mu_V)+\cE(\lambda_1)+\frac1{C^V_\BL}+v\big] 
    +N (C_1+ w(\be) ) }.
  \]
  Since $1-N^{-2/d} \ge 1/2$ for $N \ge 2,$ we get inequality
  \eqref{eq:concentrationBL} of Theorem \ref{th:theococo} with
  \begin{align*} 
    a & := \frac{1}{8C^V_\BL},\\
    b & := \frac12\Big(\ENE_V(\mu_V) +\cE(\lambda_1)+\frac1{C^V_\BL}+v\Big),\\
    c(\be) & := -\frac \be 2 \cE(\mu_V)- S(\mu_V)+ w(\beta).
  \end{align*}
  Taking $A:=\{\WAS_1(\hat\mu_N,\mu_V)\geq r\}$, we obtain the same inequality
  for the $\WAS_1$ metric after replacing $C^V_\BL$ by $C^V_{\WAS_1}$ in the
  definitions of the constants $a$ and $b$.
  
  Finally, observe that if one assumes $\liminf_{|x| \rightarrow \infty}
  {V(x)}/{|x|^\kappa} >0$ for some $\kappa >0$, then it is easy to derive the
  behavior \eqref{eq:cbeta} of the function $c(\beta)$ from \eqref{eq:cvbe}.  
\end{proof}

\begin{proof}[Proof of Corollary \ref{cor:lolo}] 
  Let us give the proof in the bounded Lipschitz case. If $\beta>0$ is fixed,
  then Theorem \ref{th:theococo} yields $u,v>0$ which only depend on $V,\beta$
  such that, for every $N\geq 2$ and $r>0$,
  \begin{equation}\label{eq:UBlolo}
    \dP_{V,\be}^N
    \Big(\DIST_\BL(\hat\mu_N,\mu_V)\geq r\Big)\leq \e^{-uN^2r^2+v Ng(N^{-1/d})}.
  \end{equation}
  Now, by a change of variables, we have for any $\alpha>0$,
  \[
  \dP_{V,\be}^N
  \Big(\DIST_\BL\big({\tau^{N^s}_{x_0}}\hat\mu_N,{\tau^{N^s}_{x_0}}\mu_V\big)\geq
  \alpha\Big)
  \leq 
  \dP_{V,\be}^N\Big(\DIST_\BL\big(\hat\mu_N,\mu_V\big)\geq \alpha N^{-s}\Big).
  \]
  The corollary follows  by taking 
  \[
  \alpha=\begin{cases}
    C N^{s-1/2}(\log N)^{1/2} & \mbox{ when } d=2,\\
    C N^{s-1/d} & \mbox{ when } d\geq 3,
  \end{cases}
  \]
  with $C >0$ large enough so that $c:=uC^2-v>0$. The proof of the $\WAS_1$
  case is the same.
\end{proof}

\section{Proof of Theorem \ref{th:tightness}}
\label{se:tightness}

We prove the exponential tightness for $\max_{1\leq i\leq N}|x_i|$ formulated
in Theorem \ref{th:tightness}, in a similar way as in the proof of
\cite[Prop.~2.1]{borot-guionnet}.

\begin{proof}[Proof of Theorem \ref{th:tightness}] Since the law
  $\dP_{V,\be}^N$ is exchangeable, the union bound gives 
  \begin{equation}\label{Xn0}
    \dP_{V,\be}^N(\max_{1\le i\le N}|x_i|\ge r)
    \le N\,\dP_{V,\be}^N(|x_N|\ge r).
  \end{equation}
  Isolating $x_N$ in the joint law, we get for any Borel set $B\subset\R^d$,
  \begin{multline}
    \label{eq:Xna}
    \dP_{V,\be}^N(x_N\in B)
    =\frac{Z_{V,\be}^{N-1}}{Z_{V,\be}^N}
    \int_{B}\e^{-\frac{\be}{2}NV(x_N)}
    \int\e^{-\frac{\be}{2}\sum_{i=1}^{N-1}(g(x_i-x_N)+V(x_i))}\\
    \times \dP_{V,\be}^{N-1}(\dd x_1,\ldots,\dd x_{N-1})\dd x_N.
  \end{multline}
  Since $V$ is admissible, there exists $C\in\R$ such that, 
  \[
  g(x-y)+V(x)+\frac12V(y)\geq C,\quad \forall x,y\in\dR^d.
  \]
  As a consequence, we obtain
  \begin{align}
    \label{eq:Xnb}
    \dP_{V,\be}^N(|x_N|\ge r)
    & =\frac{Z_{V,\be}^{N-1}}{Z_{V,\be}^N}
    \e^{-\frac\be2(N-1)C}\int_{\{|x_N|>r\}}\e^{-\frac{\be}{4}(N+1)V(x_N)}\dd x_N\nonumber\\
    &\leq \frac{Z_{V,\be}^{N-1}}{Z_{V,\be}^N}
    \e^{-\frac\be2(N-1)C-\frac{\be}{4} NV_*(r)}\int \e^{-\frac{\be}{4}V(x)}\dd x,
  \end{align}
  where we recall that $V_*(r)=\min_{|x|\geq r}V(x)$.  We now prove that
  \begin{equation}\label{ratioZN}
    \limsup_{N\to\infty}\frac1N\log\frac{Z_{V,\be}^{N-1}}{Z_{V,\be}^{N}}<\infty,
  \end{equation}
  which yields the theorem thanks to \eqref{Xn0} and \eqref{eq:Xnb} since
  $V_*(r)\to+\infty$ as $r\to\infty$.
  
  To prove \eqref{ratioZN}, let $L>0$ large enough so that $\{V<+\infty\}\cap
  B_L$ has positive Lebesgue measure and use that $\dP_{V,\be}^N(|x_N|\leq
  L)\leq 1$ together with \eqref{eq:Xna} to get
  \begin{multline*}
    \frac{Z_{V,\be}^N}{Z_{V,\be}^{N-1}}\\
    \ge
    Y_{V,\be} \int\e^{-(N-1)\frac{\be}{2}V(x_N)-\frac{\be}{2}\sum_{i=1}^{N-1}V(x_i)+g(x_i-x_N)}\,%
    \dP_{V,\be}^{N-1}\otimes\, \eta_{V,\be}(\dd x_1,\ldots,\dd x_{N}),
  \end{multline*}
  where $\eta_{V,\be}$ is the probability measure on $\dR^d$ with density
  $\frac1{Y_{V,\be}}\,\e^{-\frac\be2V}\IND_{B_L}$ and $Y_{V,\be}>0$ a
  normalizing constant. Applying Jensen's inequality with respect to the
  probability measure $\dP_{V,\be}^{N-1}\otimes \eta_{V,\be}$ then provides
  \begin{align*}
    &\log\frac{Z_{V,\be}^N}{Z_{V,\be}^{N-1}}\nonumber
    \ge\log Y_{V,\be}\label{eq:ZNZN-1}\nonumber\\
    & -\int\bigg((N-1)\frac{\be}{2}V(x_N)
    +\frac{\be}{2}\sum_{i=1}^{N-1}V(x_i)+g(x_i-x_N)\bigg)%
    \dP_{V,\be}^{N-1}\otimes \eta_{V,\be}(\dd x_1,\ldots,\dd x_{N})\nonumber\\
    & = \log Y_{V,\be} - (N-1)\frac\be 2\int V(x)\,\eta_{V,\be}(\dd x)\nonumber\\
    & -\frac{\be}{2}\int\left(\sum_{i=1}^{N-1}V(x_i)+g*\eta_{V,\be}(x_i)\right)%
    \dP_{V,\be}^{N-1}(\dd x_1,\ldots,\dd x_{N-1}).
  \end{align*}
  Since for an admissible $V$ we have $\int V \, \dd \eta_{V,\be}=\int
  \IND_{\{V<+\infty\}} V \, \dd \eta_{V,\be}<+\infty$ and that
  $g*\eta_{V,\be}$ is bounded from above, it is enough to show that 
  \begin{equation}\label{suppot}
    \sup_{N}\int\Big(\frac1N\sum_{i=1}^{N}V(x_i)\Big) %
    \dP_{V,\be}^{N}(\dd x_1,\ldots,\dd x_{N})<\infty
  \end{equation}
  in order to prove \eqref{ratioZN}. To do so, recall there exists by
  assumption $\varepsilon>0$ such that
  $V(x)-(2+\varepsilon)\log|x|\IND_{d=2}\to\infty$ as $|x|\to\infty$. In
  particular, if we set $q:=2/(2+\varepsilon)>0$, then $g(x-y)+\frac
  q2V(x)+\frac q2 V(y)\geq B$ for some $B\in\R$, and thus
  \begin{align*}
    Z_{q V,\be}^N 
    & = \int\e^{-\frac\be2\sum_{i\neq j}(g(x_i-x_j)+\frac q2V(x_i)+\frac q2V(x_j))}%
    \prod_{i=1}^N\e^{-\frac{\be q}2  V(x_i)}\dd x_i\\
    & \leq\e^{-\frac\be2 N(N-1)B}\left(\int\e^{-\frac{\be q}2 V(x)}\dd x\right)^N.
  \end{align*}
  Since Markov's inequality yields for any $R>0$,
  \[
  \dP_{V,\be}^N\bigg(\frac{1}{N}\sum_{i=1}^NV(x_i)>R\bigg)  
  \le \e^{-\frac{\be q}{2}N^2R}\, \frac{Z_{qV,\be}^N}{Z_{V,\be}^N},
  \]
  together with Lemma \ref{le:partifunc} this provides some constant $R_0\in\R$
  such that
  \[
  \dP_{V,\be}^N\bigg(\frac{1}{N}\sum_{i=1}^NV(x_i)>R\bigg) \le
  \e^{-\frac{\be}{4}N^2(R+R_0)}.
  \]
  This yields \eqref{suppot} since $V$ is bounded from below, and the proof
  of Theorem \ref{th:tightness} is therefore complete.
\end{proof}

{\small%
  \setlength{\bibsep}{.5em}
  \bibliographystyle{abbrvnat}%
  \bibliography{\jobname}%
} 

\end{document}